\documentclass[11pt]{article} 

\usepackage[utf8]{inputenc} 

\usepackage{geometry} 
\usepackage{graphicx} 

\usepackage{color}

\usepackage{booktabs} 
\usepackage{array} 
\usepackage{paralist} 
\usepackage{verbatim} 
\usepackage{subfig} 
\usepackage{fullpage}
\usepackage{hyperref}

\usepackage{algorithm}
\usepackage{algorithmic}

\usepackage{amsmath, amsfonts, amssymb, amsthm}
\newtheorem{theorem}{Theorem}
\newtheorem{lemma}{Lemma}
\newtheorem{definition}{Definition}
\newtheorem{corollary}{Corollary}
\newtheorem{claim}{Claim}

\newtheorem{remark}{Remark}

\usepackage{fancyhdr} 
\usepackage{sectsty}
\allsectionsfont{\sffamily\mdseries\upshape} 

\usepackage[nottoc,notlof,notlot]{tocbibind} 
\usepackage[titles,subfigure]{tocloft} 

\newcommand{\pw}{\mathcal{PW}}
\newcommand{\tw}{\mathcal{TW}}
\newcommand{\np}{\mathbf{NP}}
\newcommand{\nl}{\mathbf{NL}}
\newcommand{\cnf}{\mathbf{CNF}}

\newcommand{\nauxpda}{\mathbf{NAuxPDA}}

\newcommand{\sat}{\mathbf{SAT}}
\newcommand{\ksat}{k\textrm{-}\mathbf{SAT}}
\newcommand{\sac}{\mathbf{SAC}}

\newcommand{\sattw}{\mathbf{SAT}_{\mathsf{tw}}}
\newcommand{\satpw}{\mathbf{SAT}_{\mathsf{pw}}}
\newcommand{\conp}{\mathbf{coNP}}

\newcommand{\satisfiability}{\textsf{satisfiability}}

\title{Width-parameterized SAT: Time-Space Tradeoffs}
\author{Shiteng Chen \and Tiancheng Lou \and Periklis Papakonstantinou \and Bangsheng Tang}
\date{Institute for Theoretical Computer Science\\Tsinghua University}

\begin{document}
\maketitle

\begin{abstract}
Width parameterizations of $\sat$, such as tree-width and path-width,
enable the study of computationally more tractable and practical $\sat$ instances.
We give two simple algorithms. One that runs simultaneously in time-space  $\big(O^*(2^{2\tw(\phi)}),\; O^*(2^{\tw(\phi)})\big)$
and another that runs in time-space $\big(O^*(3^{\tw(\phi)\log{|\phi|}}),\;|\phi|^{O(1)}\big)$, where $\tw(\phi)$
is the tree-width of a formula $\phi$ with $|\phi|$ many clauses and variables.
This partially answers the question of Alekhnovitch and Razborov \cite{alekhnovich2002satisfiability}, 
who also gave algorithms exponential both in time and space, and asked whether the space can
be made smaller.
We conjecture that every algorithm for this problem that runs in time
$2^{\tw(\phi)\mathbf{o(\log{|\phi|})}}$ necessarily blows up the space to exponential in $\tw(\phi)$.

We introduce a novel way to combine the two simple algorithms that allows us to
trade \emph{constant} factors in the exponents between running time and space.
Our technique gives rise to a family of algorithms controlled by two parameters.
By fixing one parameter we obtain an algorithm that runs in time-space
$\big( O^*(3^{1.441(1-\epsilon)\tw(\phi)\log{|\phi|}}),\;  O^*(2^{2\epsilon\tw(\phi)})  \big)$, for every $0<\epsilon<1$.
We systematically study the limitations of this technique, and
show that these algorithmic results are the best achievable using this technique.

We also study further the computational complexity of width
parameterizations of $\sat$. We prove non-sparsification lower bounds for formulas of
path-width $\omega(\log|\phi|)$, and a separation between the complexity of path-width and tree-width
parametrized $\sat$ modulo plausible complexity assumptions.
\end{abstract}


\section{Introduction}\label{sec:intro}
Satisfiability ($\sat$) is the prototypical $\np$-complete problem extensively studied in theoretical and empirical works.
Previous work in SAT-solving deals with exact algorithms, special cases, heuristics, and parameterizations. 
In particular, width-parameterizations have received significant attention. Consider a formula $\phi$
in Conjunctive Normal Form (CNF), and also fix a graph describing its structure. 
Previous research gave algorithms with running \emph{time} exponential in width parameters,
e.g. tree-width, measured on this graph. 
One reason to care about this is because many real-world instances tend to have small width.
In this paper we take a further step and we study \emph{time-space} tradeoffs for width-based SAT-solvers.

In an influential paper Alekhnovitch and Razborov \cite{alekhnovich2002satisfiability} gave algorithms that work in time $2^{O(\tw(\phi))}$ and in space $2^{O(\tw(\phi))}$, where $\tw(\phi)$ is the tree-width of a CNF formula $\phi$; assume that $\tw(\phi)=\Omega(\log |\phi|)$. 
The authors state their results in terms of the branch-width of the formula, which is within a constant factor of the tree-width. They conclude:
\begin{quote}
``
The first important problem is to overcome the main difficulty of the practical 
implementation which is the huge amount of space used by width-based algorithms... Thus we ask if one can do anything intelligent in polynomial space to check satisfiability of formulas with small branch-width?
''
\end{quote}

The question raised by Alekhnovitch and Razborov is a major issue in practical SAT-solving. 
It is well-known in the SAT-solving community that in many common cases SAT-solvers abort
due to lack of space.

We devise two baseline algorithms for $\sat$ instances in CNF. 
These two algorithms will be used later on  as building blocks
of much more involved ones. One is similar to \cite{alekhnovich2002satisfiability}, 
and it runs in time $O^*(2^{2\tw(\phi)})$ and space $O^*(2^{\tw(\phi)})$. 
The other runs in time $O^*(3^{\tw(\phi)\log{|\phi|}})$ and space $|\phi|^{O(1)}$, and it is the 
first algorithm for deciding arbitrary CNF instances that runs in space polynomial 
and time related exponentially in the tree-width. 
Unfortunately this does not fully answer the \cite{alekhnovich2002satisfiability} question 
since we suffer a $\log|\phi|$ factor in the exponent of the running time. 
In fact, our work revolves around this logarithmic factor. 

~\\
\noindent
\textbf{Conjecture.} Let $\mathcal{A}$ be an algorithm for $\sat$ that runs in time $O^*(2^{\tw(\phi) \delta(|\phi|)})$.
Consider CNF formulas where $\tw(\phi)=\omega(\log |\phi|)$ and $\tw(\phi)=O(|\phi|^{1-\epsilon})$, for arbitrary fixed $\epsilon<0$.
If $\delta(\phi)=o(\log |\phi|)$ then $\mathcal{A}$ uses space $2^{\Omega(\tw(\phi))}$. In particular, we cannot achieve 
simultaneously time exponential in the tree-width and space polynomial in the input length.\\

Under this conjecture, for all practical purposes it makes sense to devise algorithms 
that improve the constant in the base of the running time and space from $3$ and $2$ to constants
smaller than $3$ and $2$ respectively. 
At a more systematic level one might want to obtain time-space tradeoffs between constants 
in the exponents of the running time and space.
A significant part of our contribution regards families of algorithms that achieve such tradeoffs.

Throughout this paper we assume that the tree (or path) decompositions are given in the input. 
That is, we mod-out the computational difficulty of computing the decomposition. 
This is without loss of generality for our algorithmic results since there are 
constant approximation algorithms for computing such decompositions in time exponential in the tree-width and space polynomial
in the input length (e.g. \cite{alekhnovich2002satisfiability}).
Moreover, from a complexity theory point of view this is the ``correct'' thing to do. 
In particular, when the width decomposition is given in the input, under standard complexity 
assumptions we show that deciding $\sat$ of a given a tree decomposition of width $\mathcal{W}$
is harder than deciding $\sat$ of a given path decomposition of the same width value $\mathcal{W}$.

\paragraph{Related work}
Tree-width is a popular graph parameter introduced by Robertson and Seymour \cite{robertson1983graph,robertson1986graph}.
The smaller the tree-width of a graph, the more the graph looks like a tree (in some topological sense); 
for a graph of $n$ vertices tree-width $1$ means that the graph is a tree, whereas treewidth $n-1$ means that it is
the complete graph. 
There is a handful of hard computational problems on general graphs which become computationally easier when the input graph
is of small tree-width; see. e.g. \cite{bodlaender1993tourist} for a survey.
For $\sat$ instances the tree-width of a CNF formula is the tree-width of its associated graph: incidence graph, 
primal graph, intersection graph and so on. Among those graphs, the most general one is the incidence graph (a bipartite graph
where one side has variable-nodes and the other clause-nodes). In some sense, the tree-width value on the incidence graph  
upper bounds the tree-width value of the rest \cite{szeider03}.  
There is a vast literature (too large to concisely cite here) in empirical and theoretical studies in various width-parameterizations of $\sat$. 

Improving the constant in the basis of an exponential time algorithm is a well-established goal in 
the field of exact computation for $\np$-hard problems; see e.g. \cite{woeginger2003exact} for 
a survey on problems, algorithmic techniques, and see references within. 
In particular for $\ksat$ there is a line of work in algorithms that run in time $\alpha^n$ for $\alpha<2$; e.g. \cite{paturi1998improved,schoning1999probabilistic,woeginger2003exact,MS11}.
Somewhat related to the threshold phenomenon conjectured early in this section, 
there are vertex-ordering $\np$-hard problems which can be solved 
in time-space $\big(O^*(2^n), O^*(2^n)\big)$ and in time-space $\big(O^*(4^n), n^{O(1)}\big)$; 
e.g. \cite{bodlaender-formin-etal} and references within (an in particular \cite{KP10}).
There is no easy way to adapt these technique in our case (and thus to get rid of from the exponent the logarithmic factor).
A key property of these algorithms is that is that as smaller subproblems are created the smaller the parameter 
(\emph{number of nodes}) becomes. 
There is no obvious way to achieve this when the parameter is the \emph{width} of the formula. 
Compare this to our conjecture. 
Our conjecture is about the width of a $\sat$ instance per se - furthermore,
in the worst case $\sat$ can be exhaustively solved in time $O^*(2^n)$ and space $n^{O(1)}$.

Prior to our work, \cite{GP08} addressed the question of Alekhnovitch and Razborov.
The authors gave a combinatorially non-explicit algorithm only for the $\ksat$ problem, 
where the algorithm runs in time $2^{O(\tw(\phi)\log|\phi|)}$ and space $|\phi|^{O(1)}$. 
Due to the non-explicitness the constant in the exponent of the running time cannot be bounded in some easy way. 

\cite{papakonstantinou2009note} shows that the complexity of deciding 
path-width parameterized instances precisely corresponds to the streaming verification (in log-space)
of $\np$-witnesses. In particular, it is shown that deciding formulas with path decompositions of 
width $O(\log n)$ is complete for $\nl$ and it is asked whether the complexity of $\sat$ 
instances with tree decompositions of width $O(\log n)$ is more difficult. 

Lower bounds for deciding path-width parameterized $\sat$ can be easily derived under the 
Exponential Time Hypothesis (ETH) and the application of the Sparsification Lemma \cite{IPZ01}. 
For the more general case of Constraint Satisfaction Problems, ETH has been applied in 
technically beautiful developments to essentially show that the time-optimal  
results are the standard tree-width based algorithms; see e.g. \cite{grohe,marx2010can}. 

The last question on the inherent complexity of width-parameterized $\sat$ instances 
regards their sparsification. In a model where a polynomial time verifier is given a formula $\phi$ of 
pathwidth $\pw(\phi)$ and it communicates with an all-powerful oracle, how many bits can the 
verifier send to the oracle to decide $\phi$? This question has been addressed before (e.g. \cite{FS08,DM10} - 
see \cite{DM10} for references) for $\np$-hard problems and in particular for $\sat$. In particular,
for $3$-$\sat$ \cite{DM10} conditionally shows 
that if the verifier and the oracle communicate using $n^{3-\epsilon}$ bits,
then this is not sufficient to decide satisfiability.

\paragraph{Our contribution and techniques}
We give a dynamic programming (DP) algorithm for $\sat$ where given 
a tree decomposition of the incidence graph of width $\tw(\phi)$
runs in time-space $\big(O^*(2^{2\tw(\phi)}),\; O^*(2^{\tw(\phi)})\big)$, 
and a recursive algorithm that runs in 
time-space $\big(O^*(3^{\tw(\phi)\log{|\phi|}}),$ $|\phi|^{O(1)}\big)$ (Section \ref{sec:alg}). 
The latter algorithm is the first space-efficient algorithm for width-parameterized
$\sat$. In some sense, we are doing even harder work than \cite{alekhnovich2002satisfiability}, 
since the underlying graph in that paper is the primal graph.
If we combine the DP and the recursive algorithms in the obvious way,  
then we gain the worst of both worlds.
Here ``obvious'' means that we discretize the space of truth assignments 
during the execution of the recursive algorithm and combine using DP. 
Instead, we introduce an implicit infinite family of proof systems. 
We use two free parameters to specify an algorithm in this family. 
One parameter is an integer greater than $2$. 
This controls the ``complexity'' of the rules
applied, for performing an unbalanced type of recursion of some sort. 
The larger this parameter is the more space and the smaller the 
running time is. 
The second parameter is a real number in $(0,1)$ that 
controls the discretization of the truth assignment space. 
This family of algorithms is presented in Section \ref{sec:hybrid}. In the same section we 
show that all infinite pairs of values are of interest depending on 
different time-space bounds one may want to achieve. 

Section \ref{sec:complexity} contains some preliminary complexity theory results for width-parameterizations. 
We show that the problem $\sattw$, where the $\cnf$ formula is given together 
with the tree decomposition is computationally harder than the problem 
$\satpw$ where the $\cnf$ formula is given with a path decomposition of the same value. 
In particular, $\sattw$ for tree-width $\Theta(\log|\phi|)$ is harder than $\satpw$ of path-width $\Theta(\log |\phi|)$, 
unless $\nl=\sac^1$, a standard complexity assumption (e.g. \cite{BCDRT}). 
Note that this is not true in general for other width parameters.
For example, although path-width and band-width combinatorially may be off by an exponential,
under log-space transformations they behave the same \cite{GP08}. 
We also show that there is no trivial way to sparsify $\sattw$ unless a scaled and 
non-uniform version of $\np \neq \conp$ fails. 

\section{Preliminaries}\label{sec:preliminary}
We introduce notation, terminology, and conventions used throughout the paper. We also provide a rather elementary introduction on how an algorithm may exploit the structure of bounded tree-width formulas.

\subsection{Notation}
All logarithms are of base $2$, and all propositional formulas are in Conjuctive Normal Form (CNF).
$\sat$ is the decision problem where given an arbitrary CNF formula we want to decide if
it is satisfiable. $\ksat$ denotes the restriction of $\sat$ to CNFs where each clause has at most
$k$ literals. For a formula $\phi$, $m$ denotes the number of clauses, $n$ the number of variables,
and $C_i$ and $x_j$ stand for the $i$-th clause and $j$-th variable respectively.
For convenience we write $|\phi|=m+n$. The notation $O^*$, $\Omega^*$ and $\Theta^*$ suppresses polynomial factors.

\subsection{Tree-Width}

\begin{definition}
Let $G=(V,E)$ be an undirected graph. A \emph{tree decomposition} of $G$ is a tuple $(T,X)$, where $T=(W,F)$ is a tree, and $X=\{X_1,\cdots,X_{|W|}\}$ where $X_i\subseteq V$ s.t.
\begin{enumerate}[(1)]
\item $\cup_{i=1}^{|W|}X_i=V$
\item $\forall (i,j)\in E$, $\exists t\in W$, s.t. $i, j\in X_t$.
\item $\forall i$, the set $\{t:i\in X_t\}$ forms a subtree of $T$.
\end{enumerate}
each of $X_i$ is called a \emph{bag}, the width of $(T,X)$ is defined as $max_{t\in W}|X_t|-1$, and the \emph{tree-width} $\tw(G)$ of graph $G$ is defined as the minimum width over all possible tree decompositions.
\end{definition}

When the tree decomposition $T=(W,F)$ is restricted to a path, the decomposition is called \emph{path decomposition}, and the specific tree-width is called path-width $\pw(G)$. The following inequality holds(\cite{bodlaender1998partial})
\begin{eqnarray*}
\tw(G)\le\pw(G)\le O(\log{|V|}\tw(G))
\end{eqnarray*}

\begin{definition}
The \emph{incidence graph} $G_{\phi}$ of a $\sat$ instance $\phi$ is a bipartite graph, where in one side of the bipartization each node is associated
with a distinct unsigned variable, and in the other each node is associated with a clause. There is an edge between a clause-node and a
variable-node if and only if the variable appears in a literal of the clause. The tree-width of a formula $\phi$ is the tree-width
of its incidence graph, $\tw(\phi)=\tw(G_{\phi})$. \emph{When it is clear from the context we may abuse notation and write $\tw(\phi)$
to denote the width of a given decomposition of $G_\phi$}.
\end{definition}

We assume that a tree decomposition of the incidence graph of $\phi$ is given as input along with $\phi$. For convenience, we assume the input tree decompositions have the following two properties.
\begin{enumerate}[(1)]
\item $|W|=O(\tw(\phi)\cdot|V|)=O(\tw(\phi)|\phi|)$
\item The tree $T$ has bounded degree $3$.
\end{enumerate}
Tree decompositions satisfying the two properties are called
\emph{nice}. A tree decomposition can be converted to a nice one in
linear time(\cite{kloks1994treewidth}\cite{bodlaender1998partial}). Notation $d$ is used to denote the maximal degree in the tree decomposition. By the property above, $d\le 3$. When the input is given with a path decomposition, $d$ is actually upper bounded by $2$.

\begin{remark}
The parameter $d$ affects the performance of our algorithm significantly, to fully exploit the structure of the input decomposition, we prove most of our results parameterized by $d$. One may replace it by $2$ or $3$ when the structure of the input decomposition is guaranteed to be a path or a tree.
\end{remark}

\subsection{Truth assignments, assignments, and tree decompositions}    \label{subsec:warmup}
The structure of a tree decomposition is associated with the concept of separability(e.g. \cite{bodlaender1998partial}). Intuitively, the smaller the tree-width is, the easier the graph can be broken into separate components by removing nodes. It is the separability that allows us to device more efficient algorithms for small tree-width $\sat$, than for general $\sat$. In some sense, the given tree decomposition allows us to ``localize'' the exhaustive search. The following example sheds some light on how this can be done. For the sake of simplicity, we make an additional assumption on the tree decompositions given in the input, that all the variables of a clause appear in the same bag with the clauses. We will see later that removing this assumption is non-trivial.

Suppose $x_i$'s, $x'_i$'s and $x''_i$'s are different sets of variables, and the tree decomposition  is as in Figure~\ref{subfig:simple_example}. Some clauses depending only $x'_i$'s or only $x''_i$'s are not drawn explicitly but are placed in the bags as indicated.

\begin{figure}[ht]
\centering\subfloat[Input tree decomposition.]{\includegraphics[height=140pt]{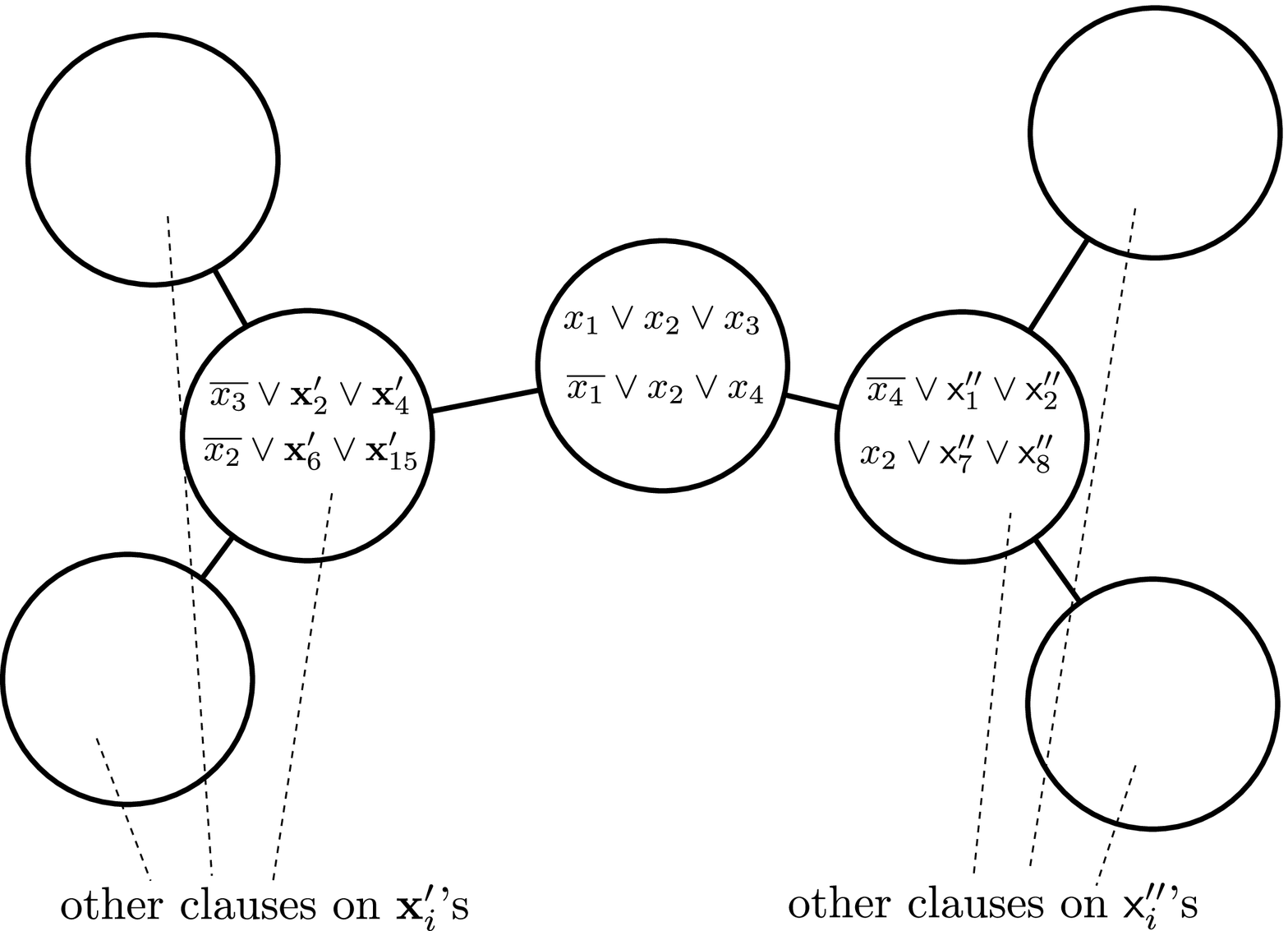}\label{subfig:simple_example}
}\hspace{40pt}\subfloat[Fixing an assignment to the variables in the middle bag results in two independent instances.]{\includegraphics[height=140pt]{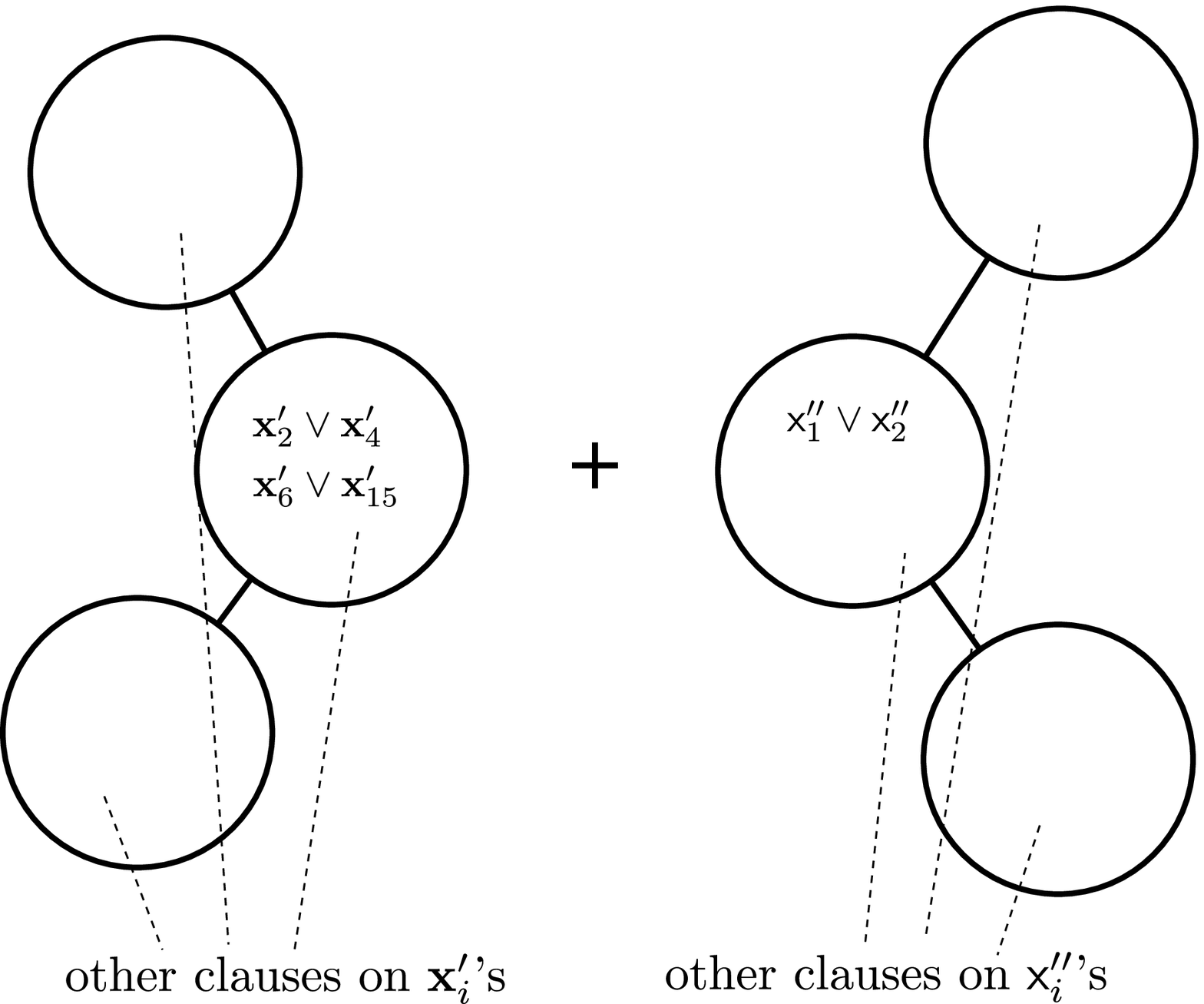}\label{subfig:after_splitting}
}

\caption{An example showing bounded tree-width $\sat$ can be solved efficiently}
\label{fig:simple_example}
\end{figure}

Suppose that we fix a truth assignment to the variables in the bag in the middle, e.g. $x_1=x_2=x_3=x_4=1$. Conditioned on this truth assignment, we can simplify the instance by removing clauses that are already satisfied, and removing literals in a clause which are set to false. This will result-in multiple subproblems as shown in Figure~\ref{subfig:after_splitting}. Assured by the property of a tree decomposition, the subproblems depend on different set of variables, i.e. they are \emph{independent}. Since if instead they shared a common variable, this variable must also wuld have appeared in the middle bag, e.g. $x_2$. At this point this variable must have been fixed to a truth assignment, thus removed in the simplification procedure.

The satisfiability of the input instance, conditioned on the truth assignment given to the middle bag, is determined by the satisfiability of the two separate subproblems. Therefore, it suffices to enumerate all truth assignments satisfying all the clauses in the middle bag without causing empty clauses in the simplification phase. Then, solve the two resulting subproblems separately to decide the satisfiability of the original instance. Furthermore, this ``splitting'' operation can be invoked recursively, by carefully choosing the ``middle'' bag.

In each recursive step, the most time-consuming part is to enumerate all the assignments satisfying all the clauses in the chosen bag, which costs $O^*(2^{\tw(\phi)\log{|\phi|}})$ time, and the total running time is $O^*(2^{\tw(\phi)\log{|\phi|}})$, which is much better than the currently best algorithms for general $\sat$ which run in the exponential in $n$. The algorithm described above will be formalized as the space-efficient algorithm in Section~\ref{subsec:space-efficient}.

\paragraph{The subtle additional assumption}
The assumption that all variables of a clause appear in the same bag with the clause is not a mild one (especially for CNFs of large cardinality). In general, we may have to delay the decision to satisfy a clause. In the above algorithm, we only store the truth assignments to the variables. The following example shows that only storing this information is not enough, when aiming at removing the assumption.

\begin{figure*}[h]
\centering{\subfloat[$\phi_1$]{\includegraphics[height=70pt]{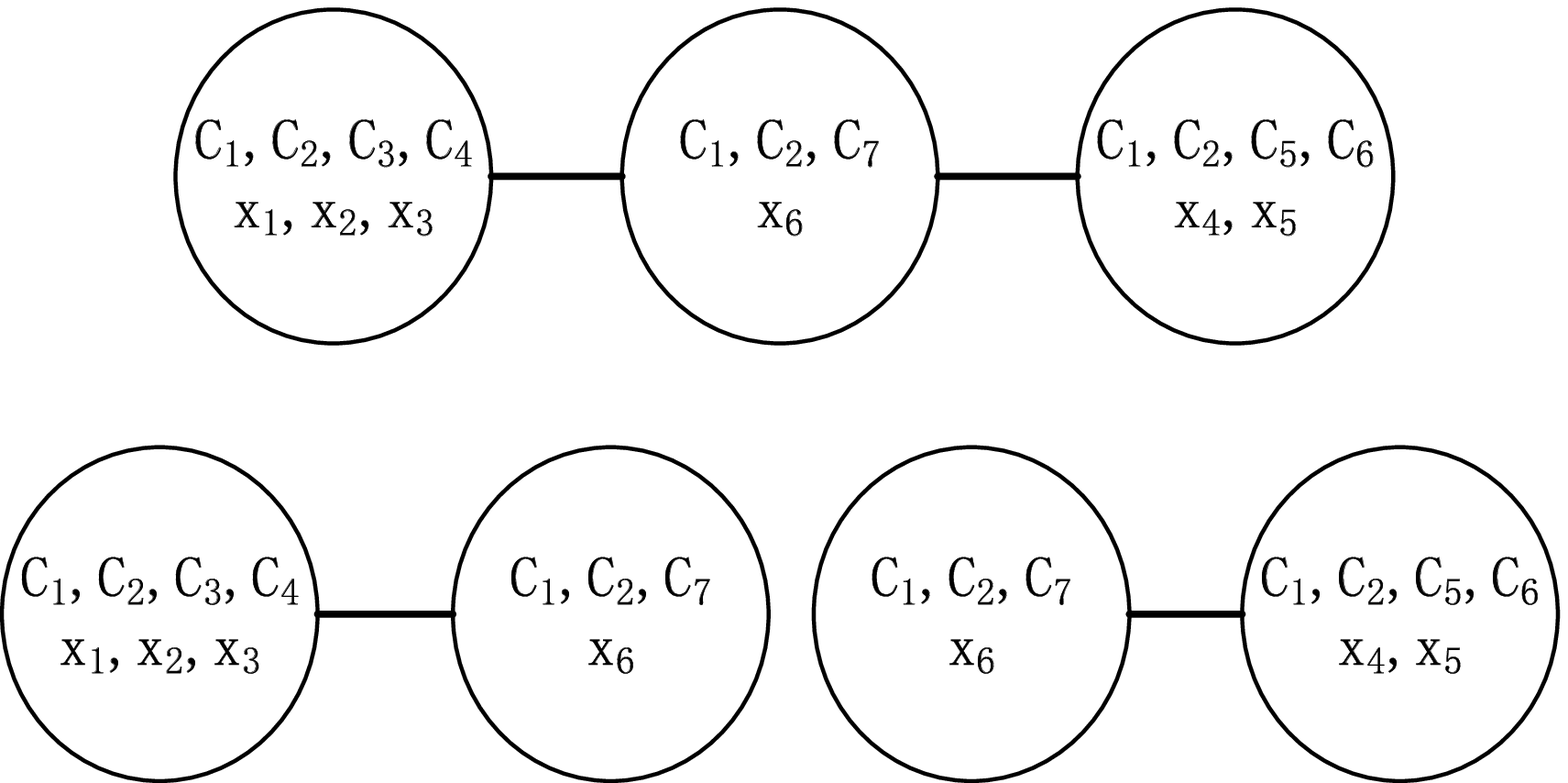}}}
\hspace{10pt}{\subfloat[$\phi_2$]{\includegraphics[height=70pt]{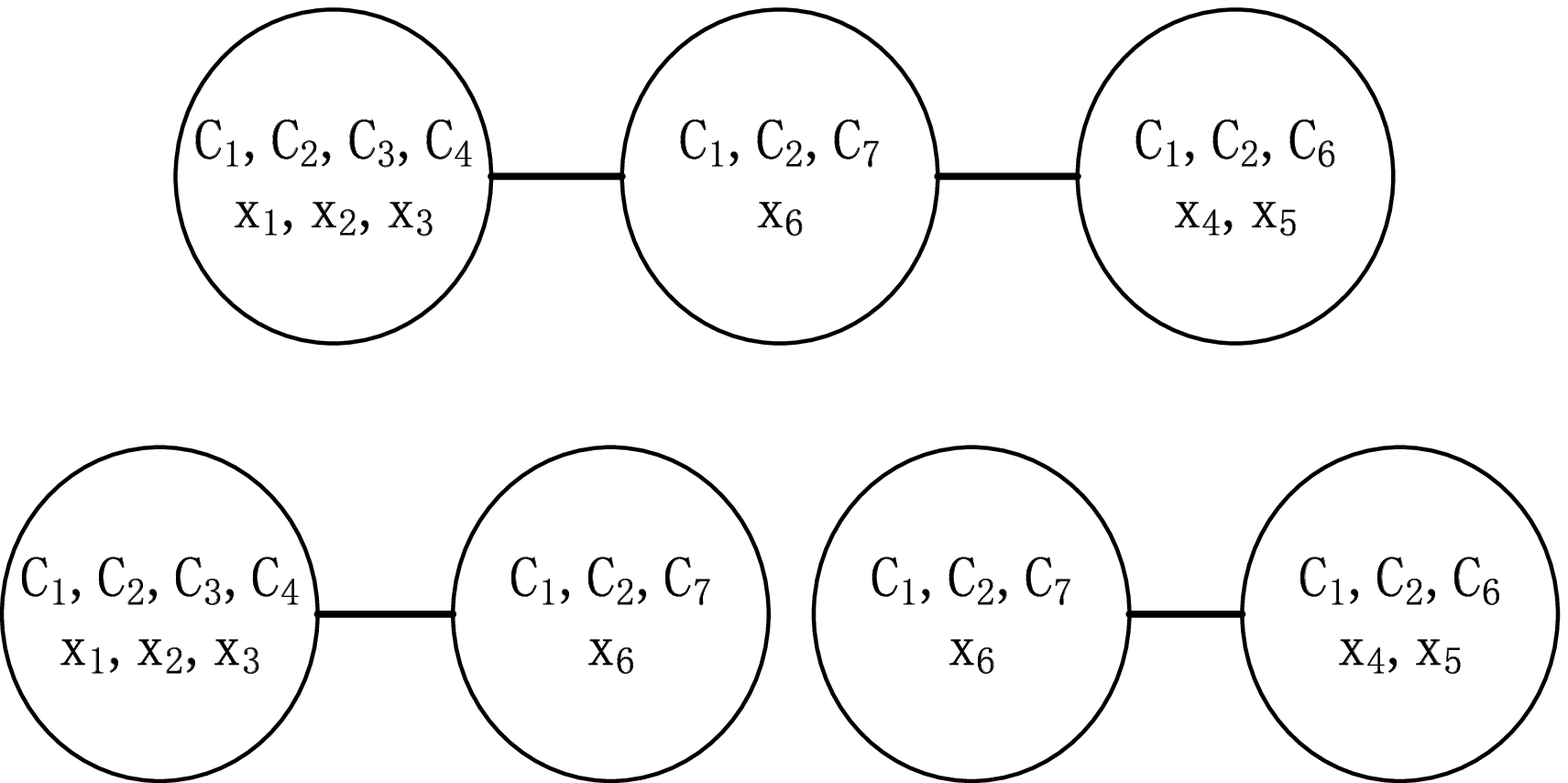}}}
\hspace{10pt}{\subfloat[$\phi_3$]{\includegraphics[height=70pt]{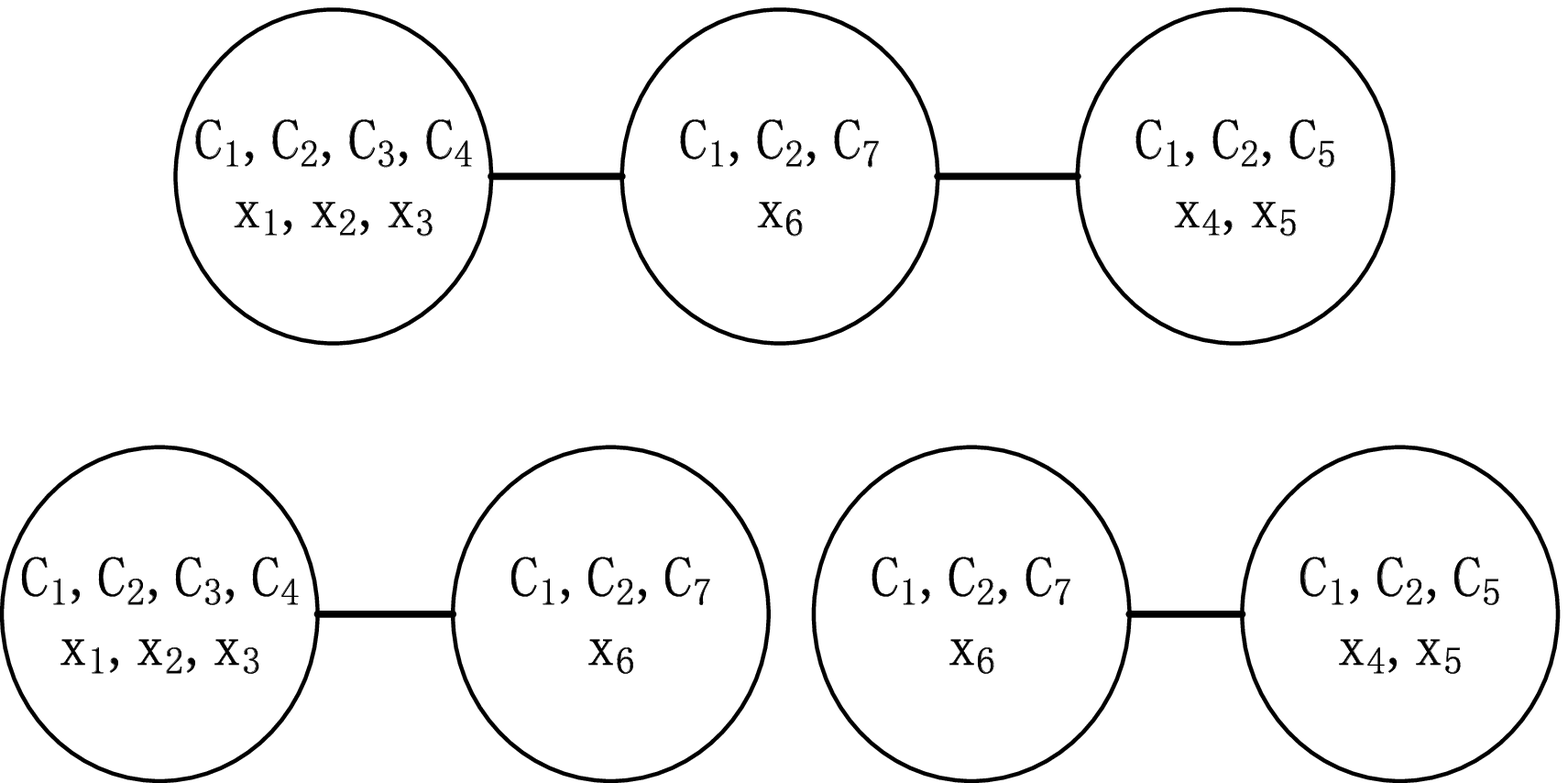}}}
\caption{Three instances used in the example. Figures on the top are
the input tree decompositions, the bottom figures are the two
components after fixing assignment to the variables in the middle
bag.}\label{fig:example_in_preliminary}
\end{figure*}

Suppose $C_1=x_1\vee x_2\vee x_4\vee x_6$,
$C_2=\overline{x_1}\vee x_3\vee x_5$, $C_3=\overline{x_2}$,
$C_4=\overline{x_3}$, $C_5=\overline{x_4}$, $C_6=\overline{x_5}$ and
$C_7=\overline{x_6}$. Three instances $\phi_1$, $\phi_2$ and $\phi_3$ along with their tree decompositions are given in Figure~\ref{fig:example_in_preliminary}, where $\phi_1=C_1\wedge\cdots\wedge C_7$, $\phi_2=C_1\wedge \cdots \wedge C_5\wedge C_7$(i.e. $C_6$ is missing), and $\phi_3=C_1\wedge\cdots\wedge C_4\wedge C_6\wedge C_7$(i.e. $C_5$ is missing). We say that a clause is satisfied by a literal under a truth assignment if the literal appears in the clause and is set to $1$. If an instance is satisfiable, then there is a truth assignment where every clause is satisfied by one of its literals.

Now, consider the splitting operation on the middle bag by fixing a truth assignment to it as above. For all three instances, the only possible assignment for $x_6$ is $0$, since $C_7$ must be satisfied by $x_6=0$. Similarly, in the left bag, we
must assign $x_2=0$ and $x_3=0$ to satisfy $C_3$ and $C_4$. In the left bag, the only variable left is $x_1$, which can satisfy either $C_1$ or $C_2$, but not both. The three instances differ in the right part, where two variables $x_4$ and $x_5$ are left.

Satisfying $C_5$ requires $x_4=0$, then $C_1$ can not be
satisfied by $x_4$. Similarly, satisfying $C_6$ requires $x_5=0$, then $C_2$ can not be satisfied by $x_5$. In order to find a satisfying truth assignment, when processing the right part, we need the information which of $C_1$, $C_2$ is already satisfied in the left part. $\phi_1$ is not satisfiable, so whichever does not affect the result. $\phi_2$ is satisfied only when $C_1$ is already satisfied, while $\phi_3$ is satisfied only when $C_2$ is already satisfied. This piece of information is not carried through the middle bag by just the truth assignment to the variables.

To overcome this issue we are going to use ``clause-bits''. In fact, the semantics of these bits is a non-obvious issue which significantly affects the running time of our algorithms. 

\paragraph{Notation and terminology}
We introduce terminology and notation to talk about truth assignment on bags. Let $X$ be a bag in the tree decomposition, $\mathcal{V}$ be the variables and $\mathcal{C}$ be the clauses appear in $X$. Also, $n_\mathcal{V}=|\mathcal{V}|$ and $m_\mathcal{C}=|\mathcal{C}|$. An \emph{assignment} $R_X$ for $X$ is a binary vector of length $n_\mathcal{V}+m_\mathcal{C}$. The first $n_V$ bits indicate the truth values of the corresponding variables. Note that the term ``assignment'' does not correspond only to a ``truth assignment'' on the variables in $X$. It is an assignment of bit values both to variables and to clauses.

What values the last $m_C$ bits have is a subtle issue explained in Section \ref{sec:alg}. For the first, dynamic programming algorithm, things are pretty clear. However, for the space-efficient and trade-off algorithms, things become more subtle. Intuitively, a bit corresponding to a clause $C$ is $1$ if we ``have decided'' to satisfy this clause (this has to do at which part of the execution of the algorithm we are), and it is different for different algorithms.

Actually, the most straightforward way of defining the clause bits is to let it denote whether the corresponding clause ``is'' satisfied. To ensure that a clause is satisfied in one of the branches in the tree decomposition, we need to enumerate all $2^d-1$ combinations that on which branches the clause is satisfied. However, if one is interested in only in the satisfiability problem (and not e.g. in $\#\sat$) we observe that only $d$ combinations can do the job.

\section{Basic algorithmic results}\label{sec:alg}
We give the two basic algorithms. These serve as building blocks for the
algorithms in Section \ref{sec:hybrid}. The first baseline algorithm (Section \ref{subsec:dp-alg})
is doing dynamic programming, and it runs simultaneously in time-space $\left(O^*(2^{2\tw(\phi)}),O^*(2^{\tw(\phi)})\right)$.
The way this algorithm goes is standard in the literature of algorithms that
compute using a tree decomposition \cite{bodlaender1993tourist}.
The second algorithm is recursive and it runs in time-space $\left(O^*(3^{\tw(\phi)\log{|\phi|}}), |\phi|^{O(1)}\right)$.
Before presenting this space-efficient algorithm we introduce more terminology and lemmas (Sections \ref{subsec:splitting}, \ref{subsec:splitting-assign})
dealing with truth assignments on tree-decompositions. This level of generality
is not necessary if it is only used for the space-efficient algorithm. Full use of this generality
is made in Section \ref{sec:hybrid} where we give the time-space tradeoff algorithms.

\subsection{Time-efficient algorithm} \label{subsec:dp-alg}
A binary array  $\satisfiability[\cdot, \cdot]$ indexed by
$X$ and $R_X$ is defined, where $X$ is the root of the subtree in
the tree decomposition, $R_X$ is an assignment to $X$.
$\satisfiability[X, R_X]=1$ means that there exists a satisfying assignment to the subtree rooted at $X$, such that the assignment to $X$ is $R_X$.

The values of the array can be computed in a bottom-up fashion. When computing values for a bag $X$, let $X_i$'s be its children, $\satisfiability[X, R_X]$ is set to $1$, if there exist $R_{X_i}$'s \emph{consistent} with $R_X$ such that $\satisfiability[X_i, R_{X_i}]=1,\forall i$. Consistent means: bits corresponding to the same variable in $R_X$ and in $R_{X_i}$'s are the same; if a bit corresponding to a clause $C$ in $R_X$ is assigned $0$, or is assigned $1$ and $C$ is satisfied by a variable in $X$, the bits in all $R_{X_i}$'s for $C$ are set to $0$, otherwise, the bits in all $R_{X_i}$'s for $C$ are set to $1$. Note that the notion of consistency here is different than the notion introduced in the next section.

Suppose $X_r$ is the root of the tree decomposition, if $\satisfiability[X_r, R_{X_r}]=1$ for some $R_{X_r}$, then the instance is satisfiable. This can be proved by induction on depth of the tree , together with the construction of the $\satisfiability$ table and the property of a tree decomposition. $\satisfiability$ table requires $O^*(2^{\tw(\phi)})$ space. Filling the entries as described above requires
$O^*(2^{(d-1)\tw(\phi)})$ time, where $d$ is the maximum degree in the tree decomposition. Recall that we have assumed a normal form on the tre decomposition where $d=3$.

\subsection{Splitting Node}   \label{subsec:splitting}
The operation of splitting the tree at a node is an essential step for the space-efficient and
tradeoff algorithms.

\begin{definition}[splitting operation]
Let $T=(V,E)$ be a tree, and $v\in V$. \emph{Splitting $T$ at $v$} is the
following operation. Let $T_1,\dots,T_k$ be the trees after removing $v$ from $T$.
The splitting operation results-in a forest $\{v\}\cup T_1,\dots,\{v\}\cup T_k$, where $\{v\}\cup T_i$
is the subtree induced by the nodes in $T_i$ together with $v$.
We call $v$ the \emph{splitting node} of this operation.
\end{definition}

The node at which we split a tree is labelled as a \emph{splitting node}.
Given a tree $T$ together with a sequence of splitting operations results-in
a forest where each subtree in the forest in general has many splitting nodes.

\begin{figure}[h]
\begin{center}
\includegraphics[width=240pt]{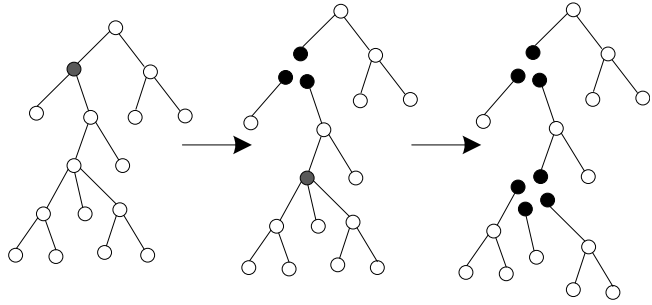}
\end{center}
\caption{Two splitting operations at the black shaded nodes}\label{fig:cutting_node}
\end{figure}

The following Lemma \ref{lemma:cutting} is somewhat reminiscent to the well-known ``$\frac{1}{3}$ - $\frac{2}{3}$
lemma'' for binary trees.
Lemma \ref{lemma:cutting} together with Corollary \ref{lemma:half},
ensures that it is possible to efficiently choose a \emph{balancing splitting node}
in a tree of constant degree;
i.e. a splitting node where the sizes of the trees in the resulted forest is linear in the
number of nodes of the tree.

\begin{lemma}\label{lemma:cutting}
Consider a tree of size $N$, a leaf $s$ and $0<\alpha<1$. Then,
there is a node $p$, where one of the trees containing $s$ resulted
after splitting at $p$ is of size $\leq \lceil\alpha N\rceil$ and
each of the rest is of size $\leq\lceil(1-\alpha) N\rceil$. $p$ is called an $\alpha$-splitting node.
Furthermore, such a $p$ can be found in time polynomial in $N$.
\end{lemma}

\begin{proof}
Here is an algorithm for finding $p$. Root the given tree at $s$
and construct a path $\langle s\equiv v_1,v_2,\dots,v_l \rangle$ as follows. At step $i$,
among the children of $v_{i-1}$ let $v_1$ be the root of the largest
subtree. We claim the there in $v_j$ in this path with the desired properties.

Denote by $a_i$ the size of the subtree containing $s$ after splitting at
$v_i$. It is obvious to see that $a_1=1$, $a_l=N$, and $a_i$
strictly increases as $i$ increases. Therefore, there must be a $j$,
such that, $a_j\le\alpha N$ and $a_{j+1}>\alpha N$. We claim that
$v_j$ is the node we need. If $a_{j+1}-a_j=1$, then cutting $v_j$
will result in two components, where the size of the component containing
$s$ is $\lceil\alpha N\rceil$, while the other one is of size
$\lceil (1-\alpha)N\rceil$. If $a_{j+1}-a_j>1$, then there must be a
branch at $v_j$, meaning that $v_j$ has at least two children.
Splitting at $v_j$ results-in at least three components, the one
containing $s$ is smaller than $\alpha N$, and the largest one among
the others is smaller than $(1-\alpha)N$.
\end{proof}

\begin{corollary}       \label{lemma:half}
On a bounded-degree tree of size $N$, there exists a node $p$, such that after splitting at $p$ each subtree is of size at most $\lceil N/2\rceil$.\end{corollary}

\subsection{Splitting nodes and assignments}   \label{subsec:splitting-assign}
Consider a tree decomposition and a sequence of splitting operations. This process breaks the original tree to a forest where each subtree ha its splitting nodes. \emph{We refer to an assignment on a subtree as the assignment that corresponds only to its splitting nodes}.
Let $\mathcal{T}$ be a subtree with splitting nodes $S$.
For a tree $\mathcal{T}$ with splitting nodes $S$, splitting at a node $p$ results in multiple subtrees $\{\mathcal{T}_i\}$, each $\mathcal{T}_i$
with its splitting nodes $S_i$.
Denote by $X^{*} = \cup_{v_i\in S}{X_i}$, and let $\mathcal{V}$ be the variables and $\mathcal{C}$ the clauses which appear in $X^*$; this is the set of variables and clauses on which we define assignments.
Suppose $R_{\mathcal{T}}$ is an assignment to $\mathcal{T}$, and $R_{\mathcal{T}_i}$ is an assignment to the subtree $\mathcal{T}_i$. $R_{\mathcal{T}}$ and all the $R_{\mathcal{T}_i}$'s are said to be \emph{consistent} if
\begin{enumerate}[(1)]
\item for every $i$, the bits corresponding to a variable $x$ in $R_{\mathcal{T}_i}$ is the same as in $R_{\mathcal{T}}$
\item for a clause $C$,
\begin{enumerate}[a)]
\item if $C$ appears in $X^*$ and is assigned $0$, then $\forall i$ every bit for $C$ in  $R_{\mathcal{T}_i}$ is assigned $0$
\item otherwise, $\exists$ exactly one $i$ such that in $R_{\mathcal{T}_i}$ the bit corresponding to $C$ is assigned $1$.
\end{enumerate}
\end{enumerate}

\begin{remark} \label{rem:exactly-one}
The latter point in the definition, where in exactly one of the subtrees we require that the corresponding bit equals to $1$, is somewhat subtle. Note that it has a significant effect in the running time of the algorithms. The following lemmas crucially depend on this issue.
\end{remark}


The following two lemmas upper bound the number of assignments in two different situations. In what follows we assume that there is an initial tree decomposition together with a sequence of splitting operations that result-in the subtrees along with their splitting nodes.

\begin{lemma}\label{lem:assignment_count}
For a tree $\mathcal{T}$ with splitting nodes $S$, the number of assignments is at most $2^{|S|\tw(\phi)}$.
\end{lemma}

\begin{proof}
$|X^{*}|\leq\sum_{v_i\in S}|X_i|\leq |S|\tw(\phi)$, the number of variables and
clauses in $X^{*}$ are at most $|S|\tw(\phi)$. So the number of
assignment is at most $2^{|S|\tw(\phi)}$.
\end{proof}

\begin{lemma}\label{lem:assignment_consistent_count}
For every assignment $\mathcal{R}_{\mathcal{T}}$ to the tree
$\mathcal{T}$, the number of assignments $R_{\mathcal{T}_i}$ to
subtrees $\mathcal{T}_i$'s consistent with
$\mathcal{R}_{\mathcal{T}}$ is at most $d^{\tw(\phi)}$.
\end{lemma}

\begin{proof}
Let $X_p$ be the bag corresponding to the splitting node $p$.
For each variable $x$ in the bag $X_p$, there are $2$
possible assignments of $x$ in the subtrees $\{\mathcal{T}_i\}$. For each clause $C$ in $X_p$, if $C$ appears in $R_{\mathcal{T}}$
and is assigned $0$, by the definition of consistency, each appearance of $C$ in the $\{\mathcal{T}_i\}$'s is
assigned $0$. Otherwise, in exactly one $\{\mathcal{T}_i\}$ $C$ is
assigned to $1$; in this case there are at most $d$ valid assignments.
\end{proof}

\begin{definition}\label{def:satisfiable}
For a tree $\mathcal{T}$ with splitting nodes $S$, an assignment $R_{\mathcal{T}}$ is \emph{satisfying} if
there exists a truth assignment $A$ to every variable in $\mathcal{T}$, such that
\begin{enumerate}[(1)]
\item every truth value for a variable in $R_{\mathcal{T}}$ agrees with the corresponding value in $A$.
\item every clause $C$ that appears in $\mathcal{T}$ where $C$ does not appear in $S$, is satisfied by $A$.
\item every clause $C$ that appears in $S$ and is assigned $1$ by $R_{\mathcal{T}}$ is such that $C$ is satisfied by $A$.
\end{enumerate}
\end{definition}

The following lemma shows how to determine the \emph{satisfiability} of an assignment recursively.

\begin{lemma}\label{thm:satisfiable}
An assignment $R_{\mathcal{T}}$ is \emph{satisfying} if and only if
there exist assignments $R_{\mathcal{T}_i}$ to the subtrees
$\mathcal{T}_i$, such that the assignments $R_{\mathcal{T}_i}$ are
consistent with $R_{\mathcal{T}}$ and each of the $R_{\mathcal{T}_i}$ are
\emph{satisfying}.
\end{lemma}

\begin{proof}
For a tree $\mathcal{T}$ with splitting nodes $S$, suppose that
splitting at node $p$ results-in the several subtrees $\{\mathcal{T}_i\}$.

Suppose that the assignment $R_{\mathcal{T}}$ is \emph{satisfying},
by Definition \ref{def:satisfiable},
there exists a truth assignment
on variables within $T$. Using the truth assignment, we can always
find assignments $R_{\mathcal{T}_i}$ consistent with
$R_{\mathcal{T}}$, such that for these truth assignments the
conditions in Definition \ref{def:satisfiable} are met.

For the other direction suppose that there exist assignments $R_{\mathcal{T}_i}$
of the subtrees $\mathcal{T}_i$, such that the assignments
$R_{\mathcal{T}_i}$ are consistent with $R_{\mathcal{T}}$ and all
$R_{\mathcal{T}_i}$ are \emph{satisfiable}. For each subtree
$\mathcal{T}_i$, there exists a truth assignment complying to
Definition \ref{def:satisfiable}. Since all these
truth assignments agree with on their common variables, we can get a truth
assignment from their union, which also meets the axioms in
Definition \ref{def:satisfiable}. Therefore, the assignment
$R_{\mathcal{T}}$ is \emph{satisfiable}.
\end{proof}


\subsection{Space-efficient algorithm}  \label{subsec:space-efficient}
The space efficient algorithm is based on the observation that in every
degree-bounded tree we can always find a node, such that by
splitting at that node every subtree has size no more than half of
the original tree (Corollary~\ref{lemma:half}).

Then, the recursive algorithm works as follows: find the
splitting node, fix the assignments for all subtrees, and then recurse
on the subtrees after the splitting. The algorithm is summarized in
Algortihm~\ref{alg:divide_and_conquer}. $\mathcal{T}$ is a tree with
previous splitting nodes $S$, and $R_\mathcal{T}$ is the assignment
fixed on the tree. A subtle point that affects the running time of this algorithm
is addressed in Remark \ref{rem:exactly-one}.

\begin{algorithm}
\caption{SAT($\mathcal{T}$, $R_\mathcal{T}$)}
\label{alg:divide_and_conquer}
\begin{algorithmic}[1]
\IF {every nodes in $\mathcal{T}$ are previous splitting nodes}

\IF {every clause in $R_\mathcal{T}$ which assigned $1$ is satisfied
by some variables in $\mathcal{T}$}

\RETURN \TRUE

\ELSE

\RETURN \FALSE

\ENDIF

\ELSE

\STATE find the splitting node $s$, and cut at $s$, which result in
many subtrees $T_{i}$

\FORALL {assignments $R_{\mathcal{T}_{i}}$ consistent with $R_\mathcal{T}$}

\IF {for each subtree $\mathcal{T}_{i}$,
SAT($\mathcal{T}_{i}$, $R_{\mathcal{T}_{i}}$) = true}

\RETURN true

\ENDIF

\ENDFOR

\RETURN false

\ENDIF

\end{algorithmic}
\end{algorithm}


Regarding its correctness, after the splitting, by the property of a
tree decomposition, all the variables shared by different components
must have been fixed at the splitting node, namely there will be no
consistency problem among the components. By induction, it can be
shown that $SAT(\mathcal{T}',\mathcal{A})$ corresponds to the
satisfiability of the (sub)tree decomposition $\mathcal{T}$
consistent with all assignments $\mathcal{A}$. Therefore,
$SAT(\mathcal{T}_0, \emptyset)$, where $\mathcal{T}_0$ is the input
tree decomposition, corresponds to the satisfiability of the
instance.


This algorithm requires only $|\phi|^{O(1)}$ space, because there are only
$O(\log{|\phi|})$ assignments to be stored during the process.
The running time can be written as follows. Suppose $T(N)$ is the
running time on a decomposition with $N$ nodes. By Lemma~\ref{lem:assignment_consistent_count}

\begin{eqnarray*}
T(N)\leq O\left( d^{\tw(\phi)}\right) T\left({1\over 2}N\right)+|\phi|^{O(1)}
\end{eqnarray*}
that is, $T(|\phi|)=O^{*}(d^{\tw(\phi)\log|\phi|})$, where by the
normal form assumption $d=3$, i.e.
$T(|\phi|)=O^{*}(3^{\tw(\phi)\log|\phi|})$.

\section{Algorithms for time-space tradeoffs} \label{sec:hybrid}
If we combine the baseline algorithms of the previous section in the
obvious way by considering blocks of truth assignments, 
we obtain the worse of the two worlds (do you see why?). In this section
we establish Theorem~\ref{the:algorithm_complexity} below, by
exhibiting a family of algorithms that achieve time-space tradeoffs.
To that end, we introduce a new algorithmic technique in which we
make non-black-box use of the two simple algorithms. Each algorithm
in this family is identified by two parameters $(\epsilon, c)$.
Moreover, we show that both of these parameters are necessary to
achieve different time-space tradeoffs. Parameter $0<\epsilon<1$
corresponds to the granularity of the discretization of the
assignment space, whereas the integer parameter $c\geq 2$ has to do
with the ``complexity'' of the rule applied recursively.
\begin{theorem}\label{the:algorithm_complexity}
For every integer $c\ge 2$ and $\epsilon$, where $0<\epsilon<1$, a $\sat$
instance $\phi$ with a tree decomposition of width $\tw(\phi)$ and
$N$ nodes, can be decided in time
$O^{*}\left(3^{\left(\lambda_{c}(\log N - c) +
c\right)(1-\epsilon)\tw(\phi)}\right)$
and in space
$O^{*}\left(2^{c\epsilon\tw(\phi)}\right)$ for a constant $\lambda_c$.
\end{theorem}

$\lambda_c$ is a constant depending on $c$. To be more specific,
$\lambda_c$ is defined as $-\log{x_c}$, where $x_c$ is the root with
largest absolute value of the polynomial equation:
$X^c-X^{c-1}-X^{c-2}-\dots-1=0$.  Values of $\lambda_c$ for small
$c$'s are listed in Table~\ref{fig:lambda_c}.

\begin{table}[h]
\centering
\begin{tabular}{|c|c|c|c|c|c|}
\hline
$c$&2&3&4&5&6\\
\hline
$\lambda_c$&1.441&1.138&1.057&1.026&1.013\\
\hline
\end{tabular}
\caption{$\lambda_c$ for small $c$'s}
\label{fig:lambda_c}
\end{table}

\subsection{Splitting Depth}
A \emph{splitting algorithm} $\mathcal{A}$ on a tree, computes a function where given a tree $\mathcal{T}$ together with previous splitting nodes $S$, it returns a node where the next splitting operation is going to be performed.

\begin{definition}\label{def:cutting_depth}
\emph{$c$-splitting depth} $\mathsf{SD}_{c}(\mathcal{A},
\mathcal{T}, S)$ of a \emph{splitting algorithm} $\mathcal{A}$ on
tree $T$ with previous splitting nodes $S$ is inductively defined as
follows:
\begin{eqnarray*}
\mathsf{SD}_c(\mathcal{A}, \mathcal{T}, S)=\left\{
\begin{array}{ll}
\max_{(\mathcal{T}_{0}, S_{0})\in C_{\mathcal{T}, S,
p}}\mathsf{SD}_c(\mathcal{A}, \mathcal{T}_{0}, S_{0}) + 1 &
|S|\le c, |S|<|\mathcal{T}|\\
0 & |S|\le c, |S|=|\mathcal{T}|\\
\infty & |S|>c
\end{array}
\right.
\end{eqnarray*}
where $p$ is the output of $\mathcal{A}$ on $\mathcal{T}$ and $S$, $C_{\mathcal{T}, S, p}$ is
the set of subtrees by splitting at $p$ in tree $\mathcal{T}$ with
previous splitting nodes $S$
\end{definition}
Intuitively, splitting depth is the recursion depth of the splitting
algorithm. \emph{$c$-minimal splitting depth}
$\mathsf{MSD}_c(\mathcal{T}, S)$ is the minimum value of
$\mathsf{SD}_c(\mathcal{A}, \mathcal{T}, S)$, over all splitting
algorithms. The case of $\infty$ is for well-definiteness.


\subsection{Assignment Group}
For a tree $\mathcal{T}$ with previous splitting at nodes $S$,
splitting a node $p$ results in several subtrees
$\{\mathcal{T}_i\}$, each $\mathcal{T}_i$ with splitting nodes
$S_i$. By Lemma \ref{lem:assignment_count}, there are at most
$2^{|S|\tw(\phi)}$ different assignments. An
\emph{$\epsilon$-assignment group} $(\epsilon$-$GR_{\mathcal{T}})$
is a set of binary strings of length at most $|S|\tw(\phi)$, in
which $(1-\epsilon)\tw(\phi)$ entries corresponding to each previous
splitting node are fixed to some constant value.

Consider the case of a tree $\mathcal{T}$ with previous splitting
nodes $S$. Suppose $S$ contains $4$ variables $x_1$, $x_2$, $x_3$,
$x_4$ and $3$ clauses $C_1$, $C_2$, $C_3$. An example of
$\epsilon$-$GR_{\mathcal{T}}$ is as follows : $x_1$, $x_4$, $C_2$,
$C_3$ are fixed to $1$,$1$,$0$,$1$, and $x_2$, $x_3$, $C_1$ are not
fixed. The $\epsilon$-$GR_{\mathcal{T}}$ contains $8 = 2^3$ binary
strings.

$\epsilon$-$GR_{\mathcal{T}}$ and
$\epsilon$-$GR_{\mathcal{T}_i}$'s are said \emph{consistent} if and
only if there exists a way to assign a value to each of the unfixed
values in $\mathcal{T}$(as $R_{\mathcal{T}}$) and $\mathcal{T}_i$'s(as
$R_{\mathcal{T}_i}$'s), such that $R_{\mathcal{T}}$ and
$R_{\mathcal{T}_i}$'s are consistent.

For a tree $\mathcal{T}$ and $\epsilon$-$GR_{\mathcal{T}}$, fix
$(1-\epsilon)\tw(\phi)$ bits correspond to variables and clauses
contained in the splitting node $p$, one can derive
$\epsilon$-$GR_{\mathcal{T}_i}$ for each subtree $\mathcal{T}_{i}$.
Note that the fixed entries for the splitting node $p$ may be
different among subtrees, and the unfixed entries in $\mathcal{T}$
need not to be fixed in subtrees. For the number of different
combinations, the following important lemma holds.

\begin{lemma}\label{lem:assignment_consistent_count_important}
The number of distinct of $\epsilon$-$GR_{\mathcal{T}_i}$'s
consistent with $\epsilon$-$GR_{\mathcal{T}}$ is at most
$d^{(1-\epsilon)\tw(\phi)}$.
\end{lemma}

\begin{proof}


For each variable $x$, there are $2$($\leq d$) possible values. For each clause $C$, let $d_0(\leq d)$ be the number of subtrees created by splitting at $p$. There are three different cases.

Suppose that $C$ does not appear in any previous splitting node. This implies that
$C$ only appears in $\mathcal{T}$, then there are $d_0$ possible
way of assigning values to the bit for $C$, s.t. exactly one of $\mathcal{T}_i$'s whose bit for $C$ is set to $1$.

Suppose that $C$ appears in some previous splitting nodes, and its value is fixed
in $\epsilon$-$GR_{T}$. If the bit for $C$ is assigned $1$, then there are
$d_0$ possible assignments to $C$ similar as above, otherwise,
the only possible way is to set all bits for $C$ to $0$.

Suppose that $C$ appears in some previous splitting nodes, but its
value is unfixed. $C$ must appear as unfixed in at least one
subtree. Without loss of generality, we assume that $C$ appears in
subtrees $\mathcal{T}_{1}, \mathcal{T}_{2}\cdots \mathcal{T}_{e_0}$,
where $e_0\geq 1$. The values of $C$ in $\mathcal{T}_{1},
\mathcal{T}_{2}\cdots \mathcal{T}_{e_0}$ are still unfixed, so there
are $d_0 - e_0 + 1\leq d_0$ possible assignments of $C$ in the
subtrees $T_{e_0 + 1}, \cdots, T_{d_0}$, the first one sets all to
$0$, and the $i(\geq 2)$-th one sets the bit of $C$ in the subtree
$\mathcal{T}_{i + e_0 -1}$ to $1$ and the rest to $0$.

Since there are at most $(1-\epsilon)\tw(\phi)$ unfixed values in
$p$, the number of different combinations of
$\epsilon$-$GR_{\mathcal{T}_i}$ consistent with
$\epsilon$-$GR_{\mathcal{T}}$ is at most $d^{(1-\epsilon)\tw(\phi)}$
\end{proof}

\subsection{Warm-up: a tradeoff algorithm for a given path decomposition}
Here we are dealing with the simpler case
where we are given a \emph{path} decomposition of the incidence graph of a
formula, together with the formula (the main complication occurs for \emph{tree} decompositions). 
Algorithm~\ref{alg:hybrid}
depicts the framework of a tradeoff algorithm, where $\mathcal{T}$
is a path decomposition, $\mathcal{S}$ is a set of previous
splitting nodes, $\epsilon$-$GR_{\mathcal{T}}$ is an assignment for
$\mathcal{T}$ with $\mathcal{S}$. This family of algorithms is
parameterized with (i) $\epsilon$ and (ii) the splitting algorithm in
line~\ref{line:splitting}. 


Let $\pw(\phi)$ be the width of the given decomposition.
The splitting algorithm chooses the node in the center
of the corresponding path segment such that the splitting operation results-in
two almost same-length segments.

Each segment will contain at most $2$ previous splitting nodes
throughout this procedure. By
Lemma~\ref{lem:assignment_consistent_count}, the number of
assignments in each segment is at most $2^{2\pw(\phi)}$. Now, we
consider a $\epsilon$-$GR_{\mathcal{T}}$ of $\mathcal{T}$. There are
at most $2^{2\epsilon\pw(\phi)}$ assignments
$R_{\mathcal{T}}\in\epsilon$-$GR_{\mathcal{T}}$. Number them from
$1$ to $|\epsilon$-$GR_{\mathcal{T}}|$. Denote by $M(\mathcal{T},
\epsilon$-$GR_{\mathcal{T}})$ a $2^{2\epsilon\pw(\phi)}$ array,
where the $i$-th entry indicates whether the $i$-th assignments of
$\epsilon$-$GR_{\mathcal{T}}$ can be \emph{satisfied}.
$M(\mathcal{T}, \epsilon$-$GR_{\mathcal{T}})$ can be computed by
Algorithm~\ref{alg:hybrid}.

This algorithm has structure similar to
the recursive algorithm, but instead of fixing
an assignment of two sub-segments it fixes only
$(1-\epsilon)\pw(\phi)$ bits and recurses on the two components with
two $\epsilon$-assignment groups. Intermediate results are stored to save time as in the dynamic-programming algorithm.

\begin{algorithm}[h]
\caption{SAT-hybrid($\mathcal{T}$, $\mathcal{S}$, $\epsilon$-$GR_{\mathcal{T}})$}
\label{alg:hybrid}
\begin{algorithmic}[1]
\STATE $M(\mathcal{T}, \epsilon$-$GR_{\mathcal{T}}) \leftarrow$ all zero matrix
\IF {all nodes in $\mathcal{T}$ are previous splitting nodes}
  \FOR {$j\leftarrow 1$ \TO $|\epsilon$-$GP_{\mathcal{T}}|$}
  \STATE Let $R_{\mathcal{T}}$ be the $j$th assignment in $\epsilon$-$GR_{\mathcal{T}}$
    \IF {all entries for a clause in $R_{\mathcal{T}}$ assigned $1$ are satisfied by some variables in ${\mathcal{T}}$}
      \STATE $M(\mathcal{T}, \epsilon$-$GR_{\mathcal{T}})_{j} \leftarrow 1$
    \ENDIF
  \ENDFOR
\ELSE
  \STATE \vspace{5pt}\framebox{\textsf{Split at the node returned by a splitting algorithm given $\mathcal{T}$, $\mathcal{S}$}}\label{line:splitting}\vspace{5pt}
  \STATE Denote the subtrees after the splitting as $\mathcal{T}_i$'s
  \FORALL {$\epsilon$-group assignments $\epsilon$-$GR_{\mathcal{T}_i}$ $\forall i$, which are all consistent with $\epsilon$-$GR_\mathcal{T}$ by fixing $(1 - \epsilon)\tw(\phi)$ entries}
    \STATE $\forall i$, $M(\mathcal{T}_i, \epsilon$-$GR_{\mathcal{T}_i}) \leftarrow\textrm{SAT-hybrid}(P, \mathcal{T}_i, \epsilon$-$GR_{\mathcal{T}_i})$
  \FOR {$j\leftarrow 1$ \TO $|\epsilon$-$GP_{\mathcal{T}}|$}
    \STATE Let $R_{\mathcal{T}}$ be the $j$th assignment in $\epsilon$-$GR_{\mathcal{T}}$
      \FORALL {$R_{\mathcal{T}_i}\in \epsilon$-$GR_{\mathcal{T}_i}$, $\forall i$}
        \IF {$M(\mathcal{T}_i,R_{\mathcal{T}_i}) = 1, \forall i$ \AND $R_{\mathcal{T}_i}$'s are all consistent with $R_{\mathcal{T}}$}
          \STATE $M(\mathcal{T}, \epsilon$-$GR_{\mathcal{T}})_{j} \leftarrow 1$
        \ENDIF
      \ENDFOR
    \ENDFOR
  \ENDFOR
\ENDIF
\RETURN $M(\mathcal{T}, \epsilon$-$GR_{\mathcal{T}})$
\end{algorithmic}
\end{algorithm}

The algorithm will recurse $\log |\phi|$ times. In each recursive
step, we only need $O(2^{2\epsilon\pw(\phi)})$ bits to store the
array $M$. Hence, we use space $O^{*}(2^{2\epsilon\pw(\phi)})$. At
every step of the recursion, we need to call the algorithm
recursively $O(2^{(1-\epsilon)\pw(\phi)})$ times. So, the running
time is $O(2^{\log {|\phi|}(1-\epsilon)\pw(\phi)})$.

The correctness follows from the correctness of the recursive
algorithm given in Section~\ref{subsec:space-efficient}. Basically,
this algorithm enumerates all the assignments as in the recursive
algorithm except that some intermediate results are stored to reduce
running time.

\subsection{A tradeoff algorithm for a given tree decomposition}  \label{subsec:hybrid}
The tradeoff algorithm for the given tree decomposition is again Algorithm \ref{alg:hybrid}, where we implement line~\ref{line:splitting} differently than in the case of a given path decomposition.

\paragraph{The main technical complication.} The idea of splitting at the center node as in the previous section does not work any more. Such splittings may create subtrees with more than $2$ previous splitting nodes. In fact, the way we dealt with this in Section \ref{subsec:space-efficient} was not to do anything. But now, we have to deal with assignment groups and thus the space requirement may be as bad as $\Omega^*(2^{c\epsilon\tw(\phi)})$, where $c$ is the maximum number of previous splitting nodes in the recursion. Overcoming this issue requires a new idea. Implementing this idea results-in an algorithm with the property that throughout its execution the number of splitting nodes in the recursive procedure is $\leq2$, or more generally a constant.

Let $\alpha$ be a parameter satisfying $0 < \alpha < 1/2$. A
$type_{\ell}$ tree is defined to be a tree with $\ell$ previous
splitting nodes. Here is a splitting algorithm $\mathcal{H}_2$
satisfying the restriction mentioned above. This specific splitting algorithm is an implementation of line~\ref{line:splitting} (i.e. replace the box 
in line~\ref{line:splitting} with Algorithm \ref{alg:line10-first}).

\begin{figure*}
\centering{\includegraphics[height=120pt]{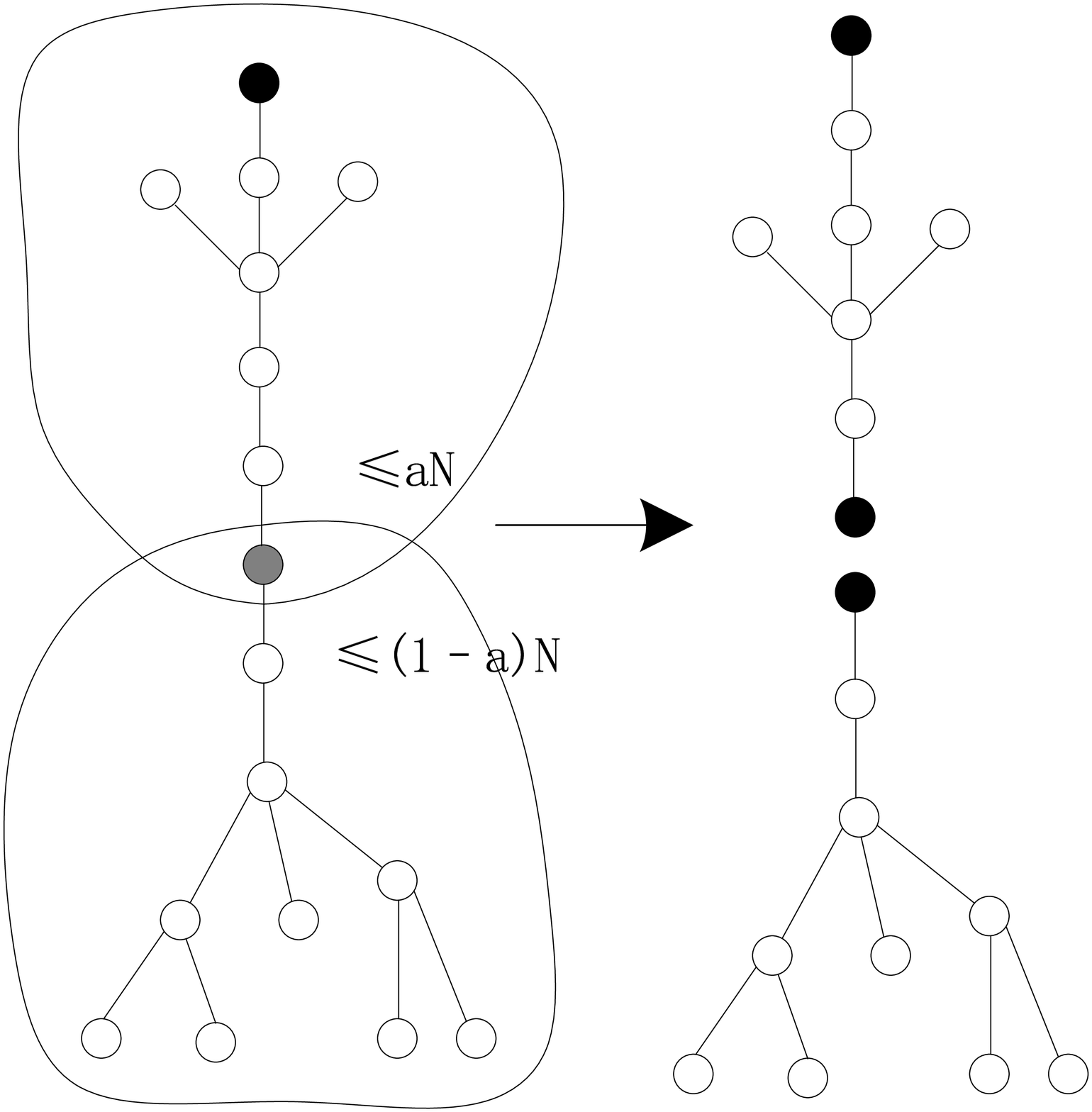}}
\hspace{20pt}{\includegraphics[height=120pt]{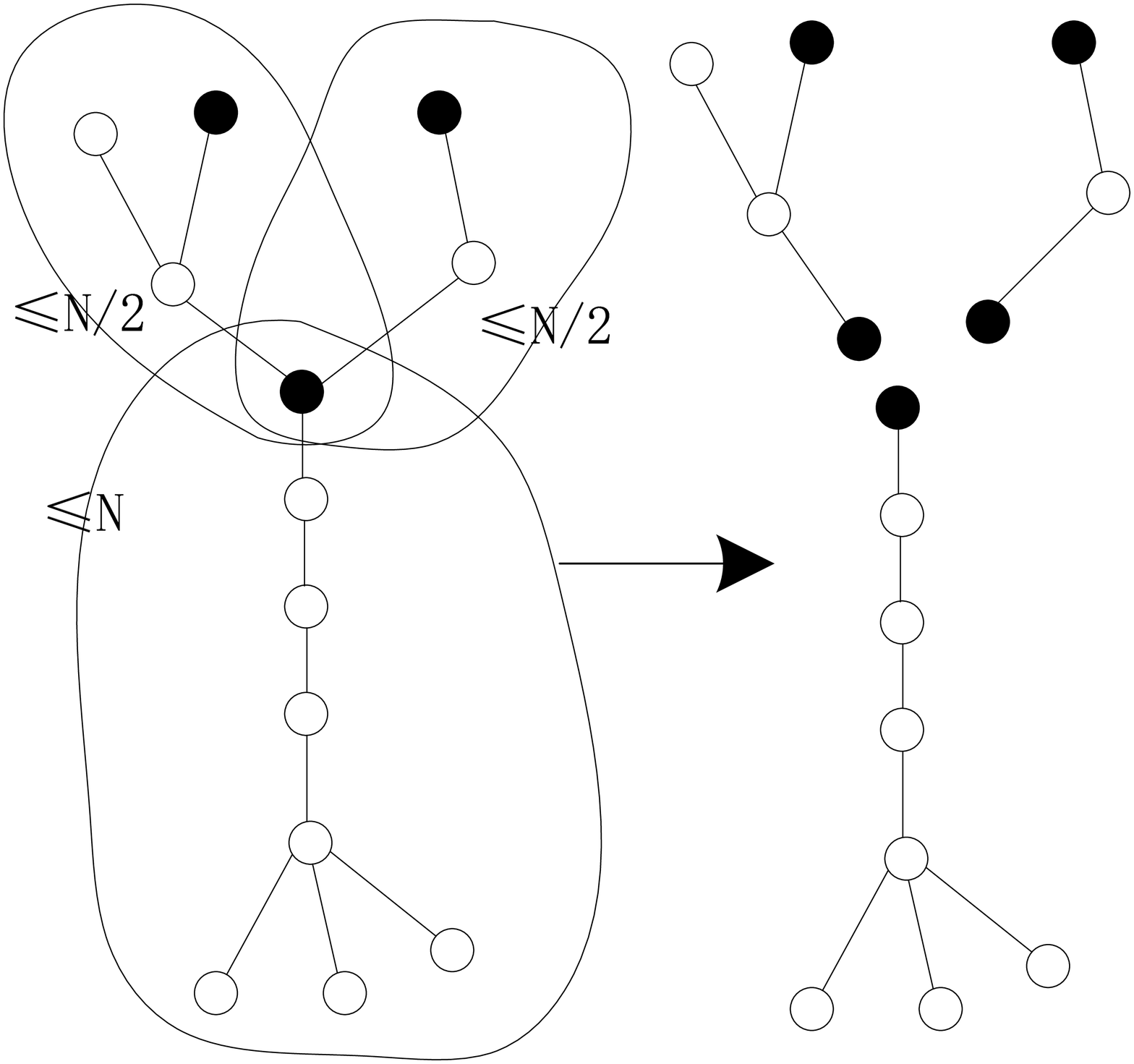}}
\hspace{20pt}{\includegraphics[height=120pt]{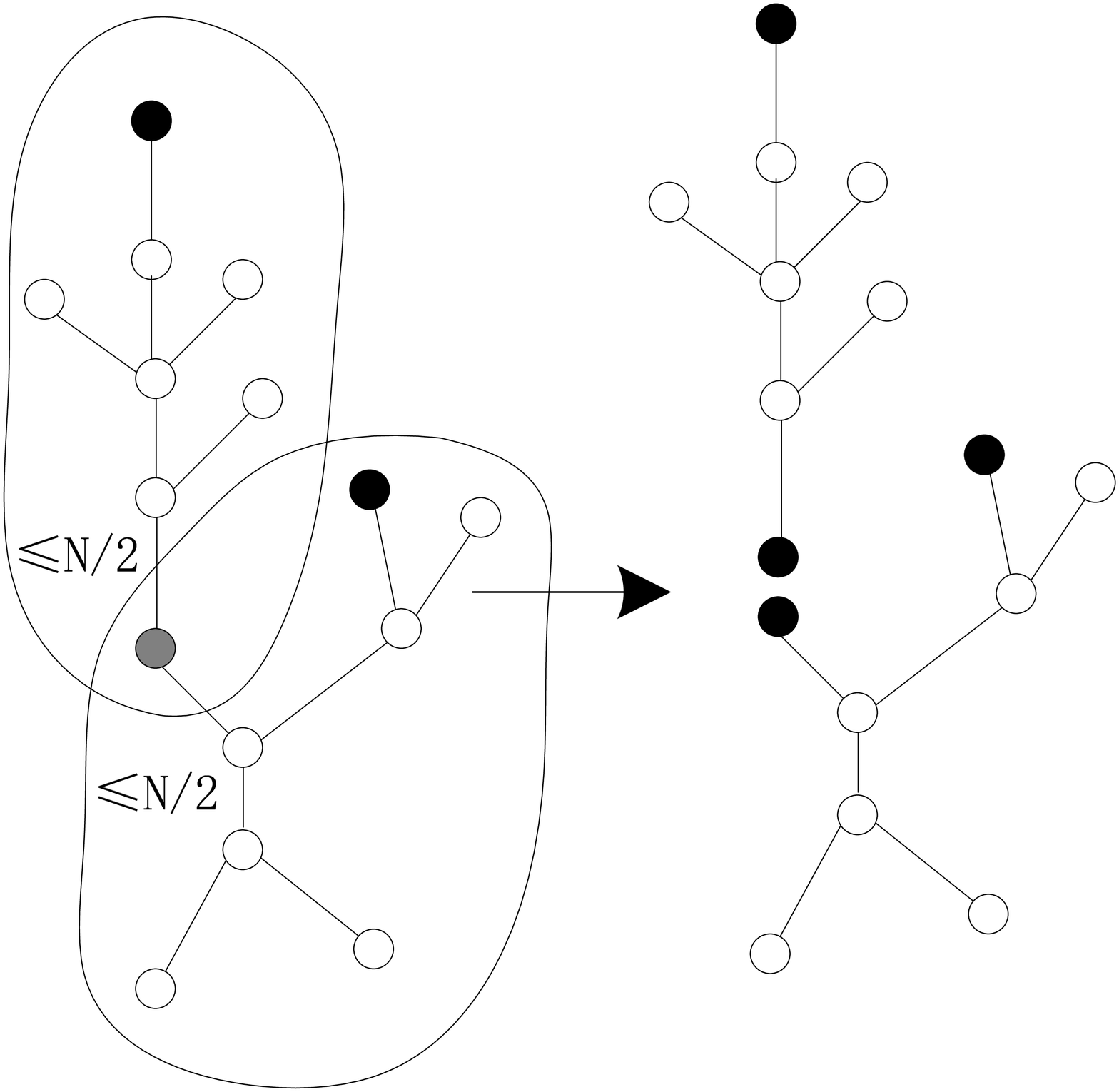}}
\caption{Finding the splitting node in three different
cases.}\label{fig:algorithm_des3_simple}
\end{figure*}

\begin{algorithm} 
\caption{Splitting algorithm $\mathcal{H}_2$ with $\mathcal{T},\mathcal{S}$ as parameters}
\begin{algorithmic}[1] \label{alg:line10-first}
\IF {$\mathcal{T}$ with $\mathcal{S}$ is a $type_{0}$ tree}
\RETURN the $1/2$-splitting node
\ELSIF {$\mathcal{T}$ with $\mathcal{S}$ is a $type_{1}$ tree}
\STATE consider the previous splitting node
as the root
\RETURN the $\alpha$-splitting node
\ELSIF {$\mathcal{T}$ with $\mathcal{S}$ is a $type_{2}$ tree}
\STATE suppose the two splitting nodes
are $p_1$ and $p_2$
\STATE consider $p_1$ as the root and compute the
$1/2$-splitting node $m$
\IF {$m$ is on the path between $p_1$ and $p_2$}
\RETURN $m$
\ELSE
\RETURN the least common ancestor
$c$ of $m$ and $p_2$
\ENDIF
\ENDIF
\end{algorithmic}
\end{algorithm}
\paragraph{Putting everything together}Utilizing this splitting algorithm, we obtain an algorithm for tree decompositions.
Denote by $T_1(N)$, $T_2(N)$ the running time of
$\mathcal{H}_2$ on $type_{1}$ or $type_{2}$ tree each of $N$ nodes respectively.

Splitting a $type_{1}$ tree will result in multiple $type_{1}$ trees
with size at most $(1 - \alpha) N$ and one $type_{2}$ tree with size
at most $\alpha N$, so we have
\begin{eqnarray*}
T_{1}(N) &\leq& O(d^{(1-\epsilon)\tw(\phi)})\left(T_{1}\left((1 -
\alpha)N\right) + T_{2}\left(\alpha N\right)\right) +
2^{O(\tw(\phi))}
\end{eqnarray*}
Splitting a $type_{2}$ tree, when the $1/2$-splitting is on the path between $p_1$ and $p_2$ will result in two $type_{2}$ trees with size at most $N / 2$ and multiple $type_{1}$ trees. Otherwise, the splitting operation will result in two $type_{2}$ trees with size at most $N / 2$ and several $type_{1}$ trees. So, we have
\begin{eqnarray*}
T_{2}(N) &\leq& O(d^{(1-\epsilon)\tw(\phi)})\left(T_{1}(N) +
T_{2}(N / 2)\right) + 2^{O(\tw(\phi))}
\end{eqnarray*}

By solving these recurrences, the running time and space of the
hybrid algorithm can be summarized as follows.
\begin{theorem}\label{the:2d-time-complexity}
$\sat$ of tree-width $\tw(\phi)$ can be solved in simultaneously
$O^{*}(d^{1.441(1-\epsilon)\tw(\phi)\log N})$ time and
$O^{*}(2^{2\epsilon\tw(\phi)})$ space, where $\epsilon$ is a free
parameter, $0<\epsilon<1$.
\end{theorem}

\begin{proof}
Set $\alpha = \frac{3-\sqrt{5}}{2}$, we have
\begin{eqnarray*}
T_{1}(N) &\leq& O(d^{(1-\epsilon)\tw(\phi)})\left(T_{1}\left((1 -
\alpha)N\right) + T_{2}\left(\alpha N\right)\right) +
2^{O(\tw(\phi))}\\
&\leq& O(d^{(1-\epsilon)\tw(\phi)})T_{1}\left((1 - \alpha)N\right) +
O(d^{2(1-\epsilon)\tw(\phi)})T_{1}\left(\alpha N\right) +
2^{O(\tw(\phi))}
\end{eqnarray*}
Therefore, we have
\begin{eqnarray*}
T_{1}(N) &\leq&
O^{*}(d^{\frac{1}{-\log{(1-\alpha)}}(1-\epsilon)\tw(\phi)\log N})
\end{eqnarray*}
Since $type_i$, $i\ge 3$ trees are not allowed, the space requirement is $O^{*}(2^{2\epsilon\tw(\phi)})$
\end{proof}

\paragraph{Optimality of the splitting procedures}
The splitting algorithm presented above is a specific one without
creating $type_i,\forall i\ge 3$ trees. Actually, it can be shown
that this specific splitting algorithm is optimal modulo our
technique.
\begin{definition}
Denote by $\mathfrak{A}_c$ ($\forall c\ge 2$) the family of
algorithms for $\sat$ with bounded tree-width following the
framework in Algorithm~\ref{alg:hybrid} which use a splitting
algorithm without creating $type_i$ trees $\forall i>c$.
\end{definition}

We lower bound the running time of all algorithms in
$\mathfrak{A}_2$. The hard instance comes from fibonacci trees.

\begin{definition}
For any positive integer $h$, a \emph{$h$-fibonacci tree}(denoted as
$F_{h}$) is recursively defined as following,
\begin{enumerate}[(1)]
\item if $h = 1$, $F_{h}$ contains only $1$ node;
\item if $h = 2$, $F_{h}$ contains $2$ nodes and one edge between them;
\item if $h > 2$, $F_{h}$ is constructed by a root connecting two subtrees $F_{h-2}$ and $F_{h-1}$.
\end{enumerate}
An \emph{extended $(h,r)$-fibonacci tree} (denote as $F^{*}_{h, r}$)
is constructed by adding one edge between the root $r$ and the root
of subtree $F_{h}$.
\end{definition}

\begin{figure}[h]
\centering{\includegraphics[height=80pt]{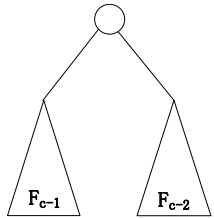}}
\caption{An \emph{$h$-fibonacci tree} ($F_{
h}$).}\label{fig:optimality_of_the_hybrid_simple}
\end{figure}

An $h$-fibonacci tree is indeed
the worst-case for splitting algorithms without creating $type_i$ trees $\forall i\ge 3$. Formally, we have the following lemma.
\begin{lemma}\label{lem:worst_case_bound_simple}
For each $h\geq 1$, in an \emph{extended $(h,r)$-fibonacci tree},
suppose $r$ is a splitting node, let $N$ be the size of the tree.
Then every splitting algorithm which does not create $type_i$ trees
$\forall i\ge 3$ runs in
$\Omega^{*}(d^{1.441(1-\epsilon)\tw(\phi)\log N})$ time.
\end{lemma}

Before getting into the proof, we define two special $type_{1}$ and
$type_{2}$ trees : $\mathcal{T}_{1,h}$ and $\mathcal{T}_{2,h}$. A
$\mathcal{T}_{1,h}$ tree is constructed by a splitting node
connected to a subtree $F_{h}$, and a $\mathcal{T}_{2,h}$ tree
constructed by two splitting nodes connected to another node which
has a subtree $F_{h}$. 

\begin{claim}\label{lem:worst_case_bound_simple_prev}
The lower bound of processing time for tree $\mathcal{T}_{1,h}$ is
$\Omega^{*}(d^{(1-\epsilon)\tw(\phi)h})$, and for tree
$\mathcal{T}_{2,h}$, the lower bound is
$\Omega^{*}(d^{(1-\epsilon)\tw(\phi)(h + 1)})$.
\end{claim}

\begin{proof}
We proceed by induction. Base cases are trivial where $h\leq 2$.
Suppose the statement is correct for any $h_0 < h$. For tree
$\mathcal{T}_{1,h}$, by inductive hypothesis, if we split at the root of $F_h$, the processing
time is at least $\Omega^{*}(d^{(1-\epsilon)\tw(\phi)(1 + (h -
1))})$, if we split at some node inside the subtrees $F_{h-1}$ or $F_{h-2}$, the processing
time is at least $\Omega^{*}(d^{(1-\epsilon)\tw(\phi)(1 + (h - 2) +
1)})$. So the lower bound for tree $\mathcal{T}_{1,h}$ is
$\Omega^{*}(d^{(1-\epsilon)\tw(\phi)h})$. For tree
$\mathcal{T}_{2,h}$, we must split at the node connecting two splitting
nodes, so again by inductive hypothesis the lower bound is $\Omega^{*}(2^{(1-\epsilon)\tw(\phi)(1 + h)})$.
\end{proof}

\begin{proof}(Proof of Lemma~\ref{lem:worst_case_bound_simple})
Set $\alpha = \frac{3-\sqrt{5}}{2}$, since the number of nodes in
$F_h$ is $\Omega\left(\left({\frac{1+\sqrt{5}}{2}}\right)^{h}\right)$, we have a lower bound for
running time
$\Omega^{*}\left(d^{\frac{1}{-\log{(1-\alpha)}}{(1-\epsilon)\tw(\phi)\log
N}}\right)$.
\end{proof}


Combining the above lemma and the algorithmic result, the following theorem can be proved, which states that our algorithm is optimal within $\mathfrak{A}_2$.
\begin{theorem}\label{the:worst_case_bound_simple}
For fixed $\epsilon$, $0<\epsilon<1$, the running time of an optimal
algorithm in $\mathfrak{A}_2$ is
$\Theta^{*}(3^{1.441(1-\epsilon)\tw(\phi)\log N})$.
\end{theorem}

\begin{figure}
\centering{\includegraphics[height=120pt]{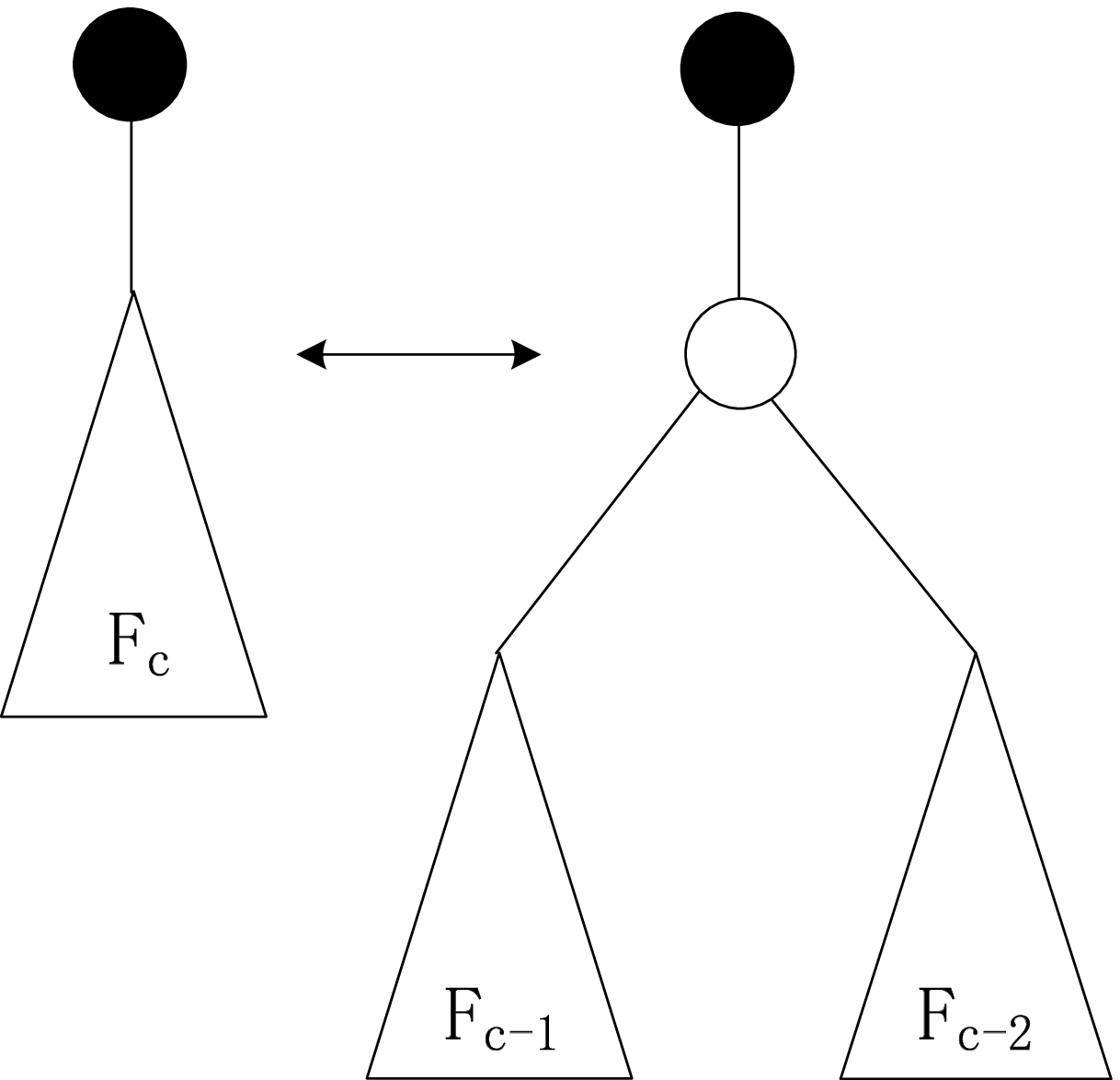}}
\hspace{40pt}{\includegraphics[height=120pt]{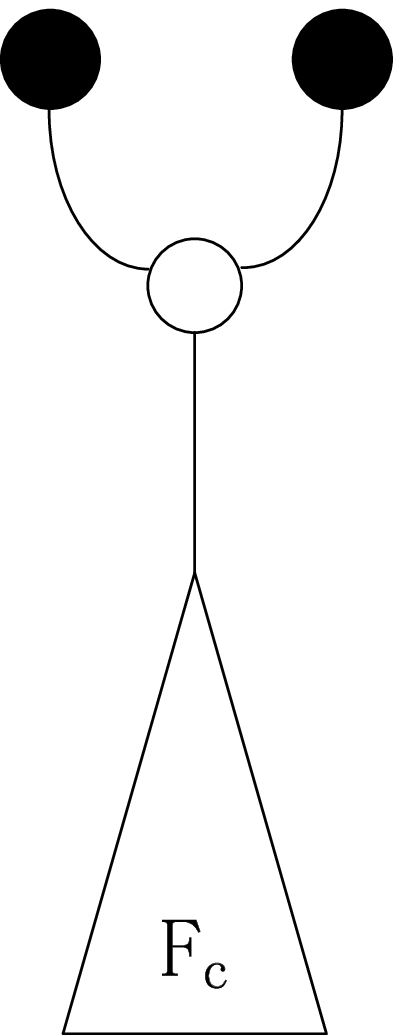}}
\caption{Illustrations of two cases for the proof of Lemma
\ref{lem:worst_case_bound_simple_prev}.}\label{fig:worst_case_proof_simple}
\end{figure}


\subsection{Generalized tradeoff algorithm for a given tree decomposition}
It is natural to ask if the algorithm can be generalized to allow up
to $type_c$ trees for some $c\ge 3$. Indeed, as $c$ increases the
running time decreases while the space requirement increases. If we want
the tradeoff to make sense, the parameter $\epsilon$ must be
restricted to some specific range.

First, we need to generalize the splitting algorithm to allow
$type_i$ trees for $i$ up to $c$. For arbitrary $1\leq i\leq c$,
consider splitting a $type_{i}$ tree: suppose the splitting node is
$p$. If $p$ is on the path between some pair of previous splitting
nodes, splitting at it will result-in several $type_{j}(j\leq i)$
trees, otherwise, splitting will result-in several $type_{1}$ trees
and one $type_{i+1}$ tree. Formally, when splitting a $type_{i}$
tree, we invoke algorithm $\mathcal{H}_c$ to determine the splitting
node, which can be seen as an implementation of line~\ref{line:splitting} in Algorithm~\ref{alg:hybrid}.

\begin{figure*}[h]
\centering{\includegraphics[height=120pt]{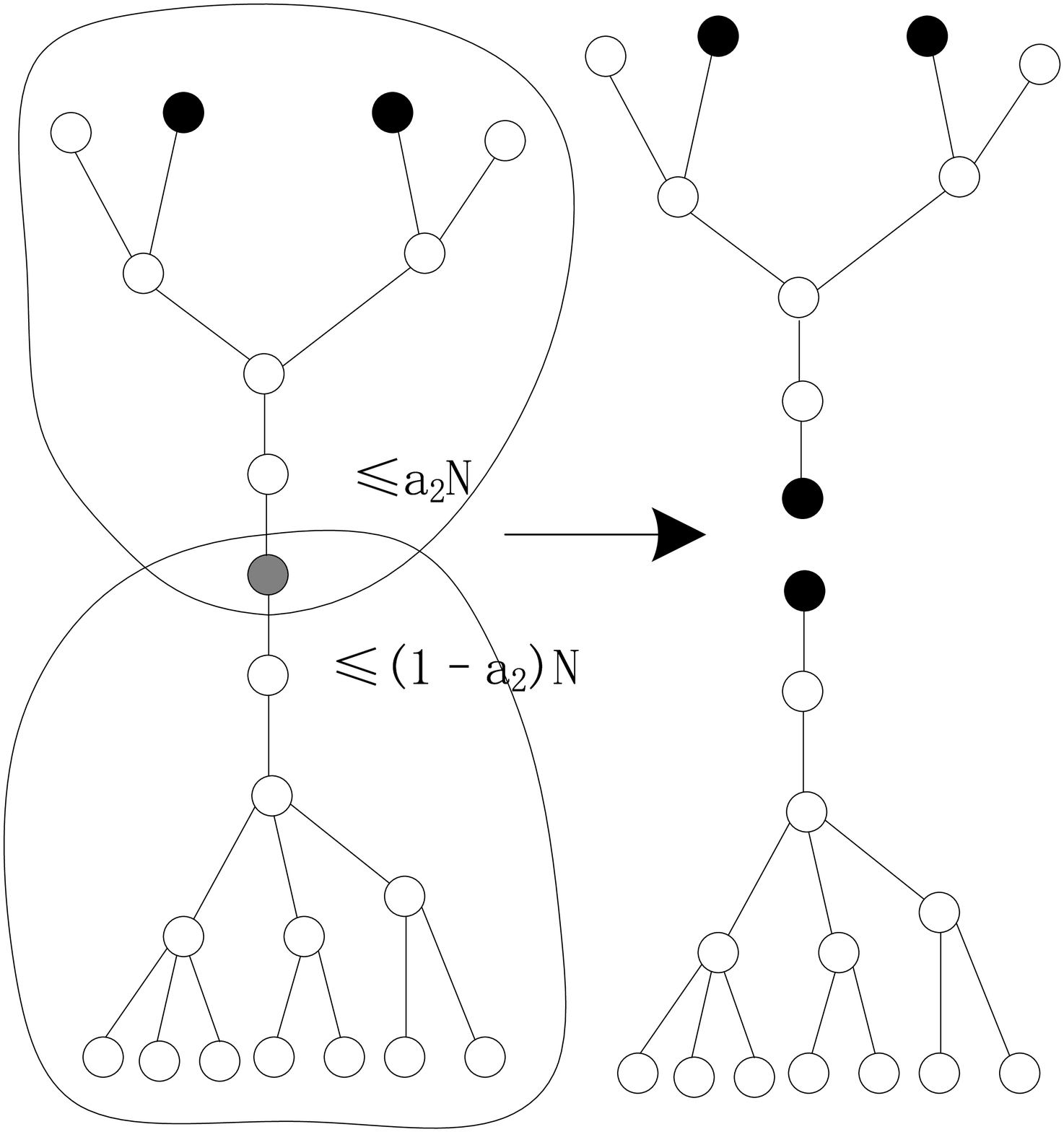}}
\hspace{20pt}{\includegraphics[height=120pt]{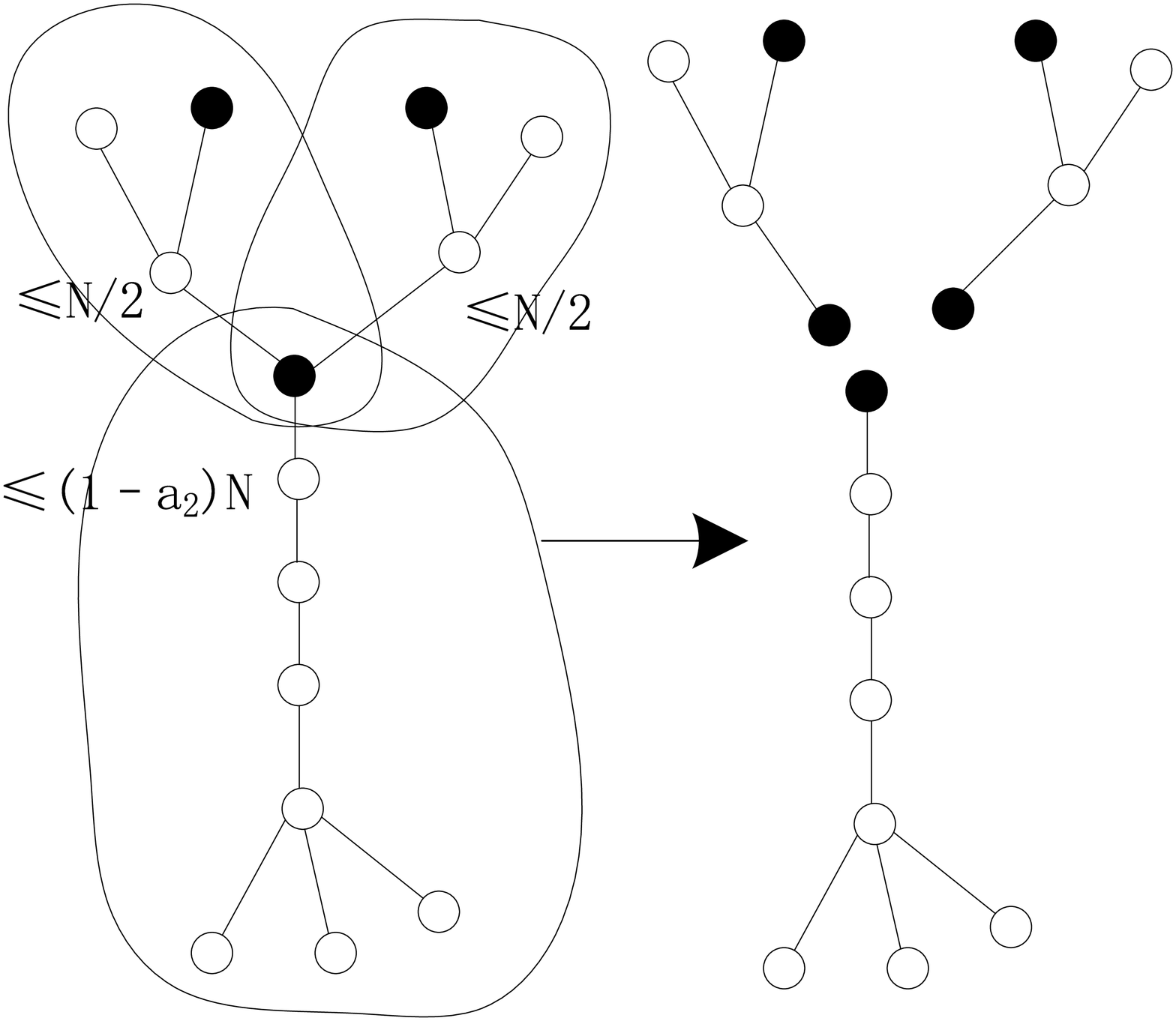}}
\hspace{20pt}{\includegraphics[height=120pt]{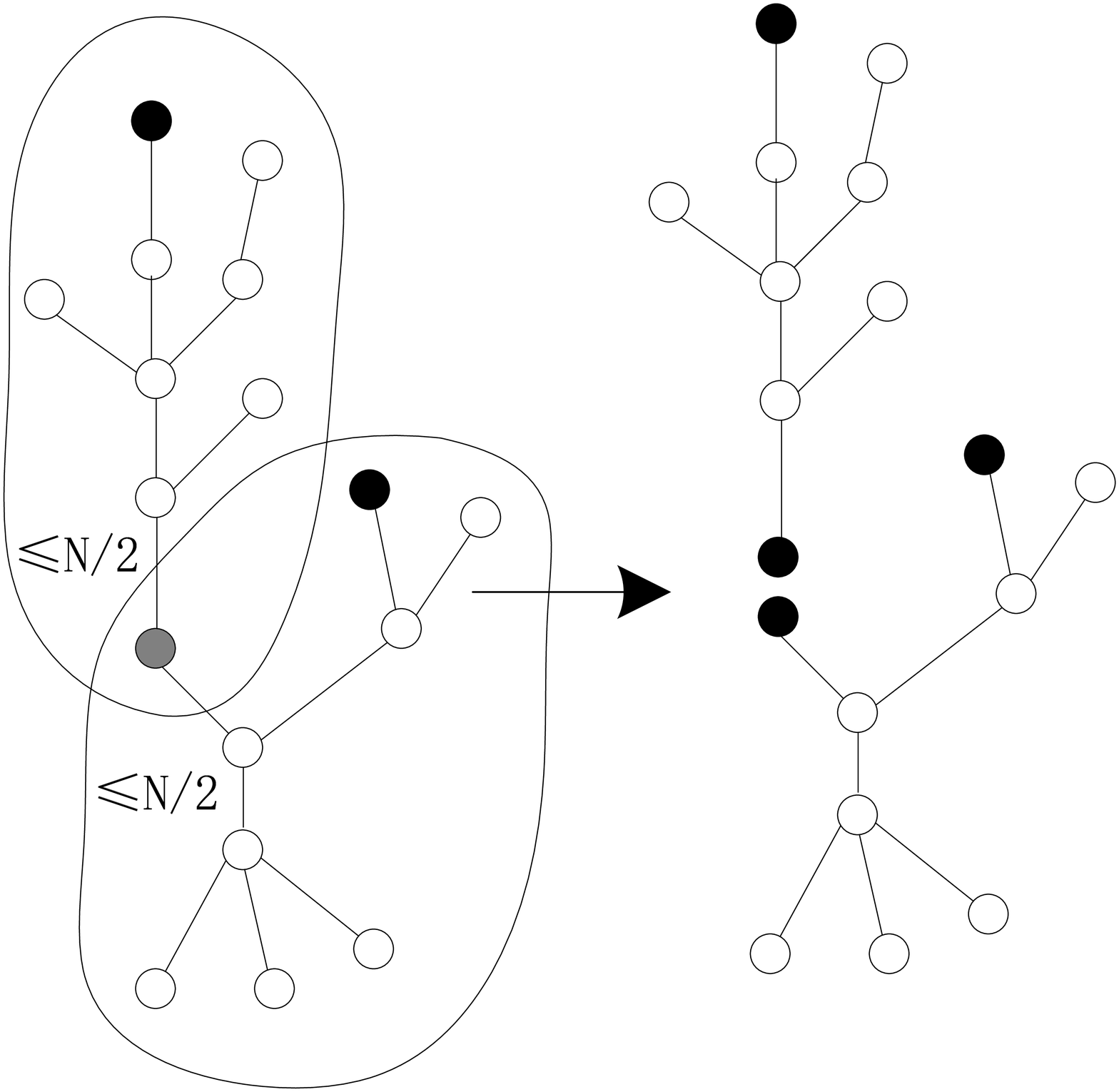}}
\caption{Finding the splitting node in three different
cases.}\label{fig:algorithm_des3}
\end{figure*}


\begin{algorithm}
\caption{Splitting algorithm $\mathcal{H}_c$ with
$\mathcal{T},\mathcal{S}$ as parameters}
\begin{algorithmic}[1]
\IF {$\mathcal{T}$ with $\mathcal{S}$ is a $type_{0}$ tree}
    \RETURN the $1/2$-splitting node
\ELSE
    \STATE suppose $\mathcal{T}$ with $\mathcal{S}$ is a $type_{i}$ tree
    \IF {the number of nodes in $\mathcal{T}$ is less than $2^{c - i}$}
        \RETURN the $1/2$-splitting node
    \ELSE
        \STATE arbitrarily pick a previous splitting node as root
        \STATE compute a $\alpha_{c,i}$-splitting node $q_1$
        \IF {$q_1$ is not on the path between any pair of previous splitting nodes}
            \RETURN $q_1$
        \ELSE
            \STATE compute a $1/2$-splitting node $q_2$.
            \IF {$q_2$ is not on the path between any pair of previous splitting nodes}
                \RETURN the least common ancestor of $q_2$ and all previous cutting nodes
            \ELSE
                \RETURN $q_2$
            \ENDIF
        \ENDIF
    \ENDIF
\ENDIF
\end{algorithmic}
\end{algorithm}

Each $\alpha_{c, i}$ for any $1\leq i< c$ is a parameter satisfying
$0\le\alpha_{c,i}\le 1 / 2$. To prevent $type_{c + 1}$ trees,
splitting nodes of $type_{c}$ trees must be on the path between some
pair of existing splitting nodes, this is assured by setting
$\alpha_{c, c} = 0$. For a fixed $c$, the running time and space of
the algorithm solving $\sat$ of bounded tree-width utilizing the
splitting algorithm $\mathcal{A}$ are summarized in
Theorem~\ref{the:algorithm_complexity}(see page~\pageref{the:algorithm_complexity}). The following lemmas are the
building blocks used to conclude the theorem.

\begin{lemma}\label{lem:algorithm_complexity}
For every $c\geq 2$, tree $\mathcal{T}$ with $N$ nodes and splitting
nodes $S$, let $D_{c, |S|}(N) =
\max_{\mathcal{T}}\{\mathsf{SD}_{c}(\mathcal{T})\}$. Then for each
$1 \leq i < c$,
\begin{eqnarray*}
D_{c, i}(N) &\leq& \max\{D_{c,1}\left((1 - \alpha_{c, i})N\right),
D_{c,i + 1}\left(\alpha_{c, i}N\right), D_{c,i}\left(N / 2\right)\}+
1
\end{eqnarray*}
and
\begin{eqnarray*}
D_{c, c}(N) &\leq& \max\{D_{c,1}(N) , D_{c,c}(N / 2)\} + 1
\end{eqnarray*}
\end{lemma}

\begin{proof}
Without loss of generality, suppose $N\geq 2^{c}$. Consider
splitting a $type_{i}$ tree with splitting nodes $S$, $1\leq i<c$.
If the $\alpha_{c, i}$-splitting-node $m$ is not on the path between
any pair of previous splitting nodes, splitting at $m$ will result
in multiple $type_1$ trees of size at most
$\lceil(1-\alpha_{c,i})N\rceil$ and one $type_{i+1}$ tree of size at
most $\lceil{\alpha_{c,i}N}\rceil$. Otherwise, since
$1-\alpha_{c,i}>1/2$, the maximal possible size of a $type_1$ tree
created by any splitting node will not exceed
$\lceil{(1-\alpha_{c,i})N}\rceil$. Splitting at the
$1/2$-splitting-node $c$ will result in multiple $type_{j}(j\leq i)$
trees of size at most $\lceil{N/2}\rceil$, otherwise, splitting at
the least common ancestor of $c$ and all previous splitting nodes as
$p$, will result in multiple $type_{1}$ tree of size at most
$\lceil{(1 - \alpha_{c, i})N}\rceil$ and many $type_{j}(j\leq i)$
trees with size at most $\lceil{N/2}\rceil$. In summary,
$$D_{c,i}(N) \leq \max\{D_{c,1}\left((1 - \alpha_{c, i})N\right), D_{c,i
+ 1}\left(\alpha_{c, i}N\right), D_{c,i}\left(N / 2\right) \} + 1$$

Now, consider splitting a $type_{c}$ tree with splitting nodes $S$.
Since $\alpha_{c, c} = 0$, we always ignore the $(1 - \alpha_{c,
i})$-splitting-node $m$. Splitting at the $1/2$-splitting-node $c$
will result in multiple $type_{j}(j\leq i)$ trees of size at most
$\lceil{N/2}\rceil$. Splitting at the least common ancestor of $c$
and all previous splitting nodes will result in multiple $type_{1}$
tree with size at most $N$ and multiple $type_{j}(j\leq i)$ trees
with size at most $\lceil{N/2}\rceil$. Namely, $$D_{c, c}(N) \leq
\max\{D_{c, 1}(N) , D_{c, c}(N / 2)\} + 1$$
\end{proof}

\begin{lemma}\label{lem:algorithm_complexity2}
For any $c\geq 2$, the $c$-splitting depth of a tree with $N$ nodes
by $\mathcal{A}$ is $\lambda_{c}(\log N - c) + c + O(1)$.
\end{lemma}

\begin{proof}

Let $D'_{c,i}(N)$ be a function satisfying the following equations,
\begin{eqnarray*}
D'_{c,i}(N) &=& D'_{c,1}((1 - \alpha_{c,i}) N) + 1 = D'_{c,i + 1}(\alpha_{c,i}N) + 1, \textrm{for $1\le i<c$}\\
D'_{c,c}(N) &=& D'_{c,1}(N) + 1
\end{eqnarray*}
For each $1< i\le c$, we have $D'_{c,i}(N) =
D'_{c,1}((1 - \alpha_{c, i - 1})N /\alpha_{c, i - 1})$. For each $1\leq i< c$, since $D'_{c,i}(N) = D'_{c,1}((1 - \alpha_{c,
i})N) + 1 $, $D'_{c,i}(N) = D'_{c,i+1}(\alpha_{c,
i}N) + 1 = D'_{c,1}(\alpha_{c, i}(1 - \alpha_{c, i + 1})N) + 2$, thus,
$$D'_{c,1}(N) = D'_{c,1}(\alpha_{c, i}(1 - \alpha_{c, i + 1}) / (1 - \alpha_{c, i})N) +
1$$ Since $D'_{c,1}(N) = D'_{c,1}((1 - \alpha_{c, 1})N) + 1$, we
have $(1 - \alpha_{c, 1})(1 - \alpha_{c, i})N = \alpha_{c, i}(1 -
\alpha_{c, i + 1})$, and $\alpha_{c, i + 1} = 1 - (1 - \alpha_{c,
1})(1 - \alpha_{c, i})/\alpha_{c, i}$. Therefore,
$$\alpha_{c, i} = 1 - \frac{\alpha_{c, 1}(1-\alpha_{c, 1})^{i}}{2\alpha_{c, 1} -
1 + (1 - \alpha_{c, 1})^{i}}$$
On the other hand, $D'_{c,c}(N) = D'_{c,1}(N) + 1 = D'_{c,1}((1 -
\alpha_{c, c-1})N/\alpha_{c, c-1})$, so we have $\alpha_{c, c-1}/(1
- \alpha_{c, c}) = 1 - \alpha_{c, 1}$ and $\alpha_{c, c-1} = (1 -
\alpha_{c, 1})/(2 - \alpha_{c, 1})$. Thus, $(1 - \alpha_{c,
1})/(2 - \alpha_{c, 1}) = 1 - \frac{\alpha_{c, 1}(1-\alpha_{c,
1})^{c - 1}}{2\alpha_{c, 1}- 1 + (1 - \alpha_{c, 1})^{c - 1}}$, and
$\sum_{i=1}^{c}(1 - \alpha_{c, 1})^{i} = 1$.

By setting $\lambda_{c} = \frac{1}{\log (1 - \alpha_{c, 1})}$, for
each $1\leq i\leq c$, $D'_{c,i}(N) \geq D'_{c,i}(N / 2) + 1$. And
then $$D'_{c,i}(N) = \max\{D'_{c,1}((1-\alpha_{c, i}) N),
D'_{c,i+1}(\alpha_{c, i}N), D'_{c,i}(N/2)\} + 1$$ and $$D'_{c,c}(N)
= \max\{D'_{c,1}(N), D'_{c,c}(N/2)\} + 1$$ Combining with the
previous lemma, it can be easily proved by induction that
$D'_{c,i}(N)$ upper bounds $D_{c,i}(N)$ for all $c,i$. Specifically,
$$D_{c,1}(N) \leq D'_{c,1}(N) \leq \lambda_{c}(\log N - c) +
D_{c,1}(2^{c})+ O(1)$$ Since $D_{c,1}(2^c) = c$, the $c$-splitting
depth of a tree with $N$ nodes by the algorithm $\mathcal{A}$ is
upper bounded by $\lambda_{c}(\log N - c) + c + O(1)$, where
$\lambda_{c}$ satisfies $\sum_{l=1}^{c}2^{-\frac{l}{\lambda_{c}}} =
1$.
\end{proof}

\begin{proof}(Proof of Theorem~\ref{the:algorithm_complexity})
For every $c \geq 2$, we solve the above recurrences where $N <2^{c}$. The running time is
$O^{*}\left(d^{\log N (1-\epsilon)\tw(\phi)}\right)$
when $N \ge 2^{c}$, the running time is
$O^{*}\left(d^{\left(\lambda_{c}(\log N - c) + c\right) (1-\epsilon)\tw(\phi)}\right)$.
The space is upper bounded by $O^{*}(2^{c\epsilon\tw(\phi)})$ since only $type_i$($i\le c$) trees are allowed.
\end{proof}

Furthermore, it can be proved that the space resource can be
fully exploited to minimize the running time, which is of practical
importance. More specifically, the following holds true.

\begin{corollary}[of Theorem~\ref{the:algorithm_complexity}]\label{cor:hybrid_opt_space}
For any $\epsilon'>0$ there exists an algorithm which runs in space
$O^{*}(2^{\epsilon'\tw(\phi)})$ and time
$O^{*}(d^{\delta\tw(\phi)\log_{2}{|\phi|}})$ time for a constant
$\delta<1$.
\end{corollary}

\begin{lemma}\label{lem:bound_of_lambda}
$\lambda_c<1+\frac{2}{2^{\frac{c}{2}}}$
\end{lemma}
\begin{proof}
Let $f(X)=X^{c}-\sum_{i=0}^{c}{X^{i}}$, then $\gamma_{c}$ is the
root of $f(X)=0$ with largest absolute value. We know $f(2)=1>0$, so
if we can prove $f(2-\frac{1}{2^{\frac{c}{2}}})<0$ then there must be a
root between 2 and $2-\frac{1}{2^{\frac{c}{2}}}$. Denote
$y=2-\frac{1}{2^{\frac{c}{2}}}$,
\begin{eqnarray*}
f(y)<0\iff y^{c}<\sum_{i=0}^{c}{y^{i}}=\frac{y^{c}-1}{y-1}\iff y<2-\frac{1}{y^c}
\end{eqnarray*}
The last inequality is true because
$y=2-\frac{1}{2^{\frac{c}{2}}}>\sqrt{2}$ when $c\ge2$ and
$2-\frac{1}{y^c}>2-\frac{1}{\sqrt{2}^{c}}=y$. By
$\lambda_c=\frac{1}{\log_{2}{\gamma_{c}}}$,
$\lambda_c<1+\frac{2}{2^{\frac{c}{2}}}$.
\end{proof}
Given the upper bound on $\lambda_c$, the corollary can be proved as
follows.
\begin{proof}(Proof of Corollary~\ref{cor:hybrid_opt_space})
For fixed $\epsilon$ and $c$, by Theorem~\ref{the:algorithm_complexity}, there is an
algorithm with running time
$O^*(d^{\lambda_{c}(1-\epsilon)\log_{2}{N}\tw(\phi)})$ and space
$O^*(2^{c\epsilon\tw(\phi)})$ for any $\epsilon>0$. Set
$\epsilon=\frac{\epsilon'}{c}$, then the space is
$O^*(2^{\epsilon'\tw(\phi)})$ and the running time is
$O^*(d^{\lambda_{c}(1-\frac{\epsilon'}{c})\log_{2}{N}\tw(\phi)})$.
By Lemma~\ref{lem:bound_of_lambda},
$\lambda_{c}(1-\frac{\epsilon'}{c})<(1+\frac{2}{2^{\frac{c}{2}}})(1-\frac{\epsilon'}{c})<1$
when we pick a large enough $c$.
\end{proof}


\begin{remark}\label{rem:time_and_depth}
Given a tree $\mathcal{T}$, any algorithm $\mathcal{A}$ avoiding
$type_{c+1}$ trees with splitting depth
$\mathsf{SD}_{c}(\mathcal{A}, \mathcal{T}, \emptyset)$ requires
$O^{*}(d^{(1-\epsilon)\mathsf{SD}_{c}(\mathcal{A}, \mathcal{T},
\emptyset)\tw(\phi)})$ time and $O^{*}(2^{c\epsilon\tw(\phi)})$
space.
\end{remark}

\paragraph{Optimality}
Similarly to the second half of the previous section, we also prove the optimality of the generalized tradeoff algorithm. We construct the hard instance using generalized fibonacci trees.

\begin{definition}
For any integer $c\ge 2$, and a positive integer $h$, a \emph{$(c,
h)$-fibonacci tree}(denoted as $F_{c,h}$) is defined as by one of
the rules,
\begin{enumerate}[(1)]
\item if $h\le c$, $F_{c,h}$ is a chain of $2^c$ nodes;
\item if $h>c$, $F_{c,h}$ is constructed by starting from a chain of $c$ nodes, then replacing the $i$th node by a subtree $F_{c,h-i}$.
\end{enumerate}
An \emph{extended $(c,h,r)$-fibonacci tree} (denote as $F^{*}_{c, h,
r}$) is constructed by connecting one root node $r$ to a subtree
$F_{c, h}$.
\end{definition}

See Figure~\ref{fig:optimality_of_the_hybrid} for an illustration of
a $(c,h)$-fibonacci tree. A $(c,h)$-fibonacci tree is indeed the
hardest input of the splitting algorithm. To be more specific, the
following lemma holds.
\begin{lemma}\label{lem:worst_case_bound}
For each $h\geq 1$, $\mathsf{MSD}(F^{*}_{c, h, r}, \{r\}) \ge h$.
\end{lemma}

\begin{proof}
For any $c\geq 2$, $h > c$ and $1 \leq w \leq c$, the tree $G_{c, h,
w}$ is defined as follows, first construct a chain of length $w$,
connect $c - w + 1$ splitting nodes to the first node of the chain,
and connect a subtree $F_{c, h - c + w - i}$ to $i$-th node of the
chain. Denote $S_{\ell}$ as the set of the $\ell$ splitting nodes
connected to the first node of the chain. We prove that
$\mathsf{MSD}(G_{c, h, w}, S_{c - w + 1})\geq h - c + w$, which
implies the inequality that we need. Specifically,
$\mathsf{MSD}_{c}(F^{*}_{c, h, r},\{r\}) = \mathsf{MSD}_{c}(G_{c, h,
c},S_{1})\geq h - c + c \geq h$.

The inequality is proved by induction on $h$. The basic case is
trivial. Suppose for any $h < h_0$, $\mathsf{MSD}_c(G_{c, h, c},S_{c
- w + 1})\geq h - c + w$. Now we prove $\mathsf{MSD}_c(G_{c, h_0,
c},S_{1})\geq h_0 - c + w$ by induction on $w$. When $w = 1$, to
prevent $type_i$ tree for $i > c$, we must split at the first node
of the chain. Therefore, $\mathsf{MSD}_{c}(G_{c, h_0, 1}, S_{c}) = 1
+ \mathsf{MSD}_{c}(G_{c, h_0 - c, c}, S_{1}) \geq h_0 - c + 1$. When
$w > 1$, if the splitting node is in the subtree $F_{c, h_0 - c + w
- 1}$ connected to the first node of the chain,
$\mathsf{MSD}_{c}(G_{c, h_0, w}, S_{c - w + 1}) \ge 1 +
\mathsf{MSD}_{c}(G_{c, h_0, w - 1}, S_{c - w + 2})\ge h_0 - c + w$,
otherwise $\mathsf{MSD}_{c}(G_{c, h_0, w}, S_{c - w + 1}) \ge 1 +
\mathsf{MSD}_{c}(G_{c, h_0 - c + w - 1, c}, S_{1}) = h_0 - c + w$.
\end{proof}

\begin{figure}[h]
\centering{\includegraphics[height=200pt]{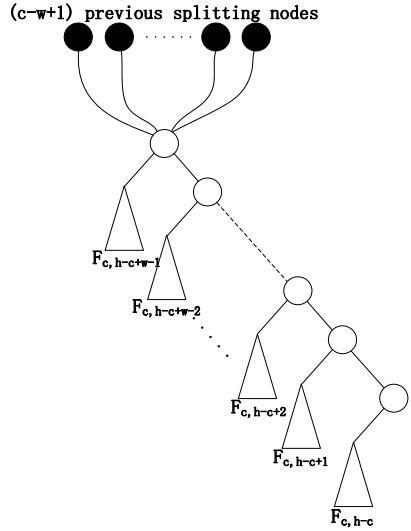}}
\caption{A tree $G_{c, h, w}$.}\label{fig:worst_case_proof}
\end{figure}

\begin{theorem}\label{the:worst_case_bound}
For every $c\geq 2$ and $N > 2^c$, there exists a tree $\mathcal{T}$
with $N$ nodes, such that the $c$-minimal splitting depth of
$\mathcal{T}$ $\mathsf{MSD}_c(\mathcal{T}, \emptyset)$ is at least
$\lambda_{c} (\log N - c) + c - O(1)$.
\end{theorem}

\begin{proof}
Let $|F_{c, h}|$ be the number of nodes in the tree $F_{c,
h}$. For any $h\leq c$, we have $|F_{c, h}|\leq 2^{c}$, when $h
> c$, we have $|F_{c, h}| = \sum_{i = 1}^{c}|F_{c, h - i}| + c$. By the recursive relation,
generating function of $|F_{c, h}|$ can be written as $f(X) = X^{c}
- \sum_{i=0}^{c}{X^{i}}$. Therefore $|F_{c, h}| = \sum_{i =
1}^{c}{\delta_{c, i}\gamma_{c, i}^{h - c}}$, where $\delta_{c,i}$ is
at most constant times of $2^{c}$ and $\gamma_{c, i}$ is the $i$-th
root of the equation $f(X) = 0$.

Let $\gamma_{c} = \arg\max_i\{|\gamma_{c,i}|\}$. When $h$ tends to
infinity, $|F_{c, h}| = \Theta(2^{c}\gamma_{c}^{h - c})$. So, $h
\geq \log_{\gamma_{c}}{\left({|F_{c, h}| / 2^{c}}\right)} + c - O(1)
= \lambda_{c}(\log{|F_{c, h}|} - c) + c - O(1)$. Therefore, for any
$c\geq 2$ and $N > 2^c$, there exists a tree $\mathcal{T}$ with $N$
nodes, such that the $c$-minimal splitting depth of $\mathcal{T}$
$\mathsf{MSD}_c(\mathcal{T}, \emptyset)$ is at least $\lambda_{c}
(\log N - c) + c - O(1)$.
\end{proof}

\begin{figure}[h]
\centering{\includegraphics[height=160pt]{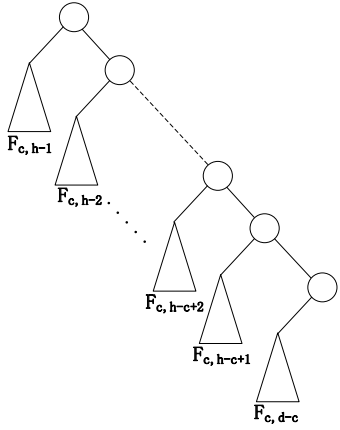}}
\caption{A \emph{$(c, h)$-fibonacci tree} ($F_{c,
h}$).}\label{fig:optimality_of_the_hybrid}
\end{figure}

Similarly to Theorem \ref{the:worst_case_bound_simple}, we conclude
the optimality of our tradeoff algorithm, namely, for fixed $c\ge
2$, $\epsilon$, $0<\epsilon<1$, any algorithm in $\mathfrak{A}_c$
runs in $\Omega^*(3^{\lambda_{c} (\log N - c) + c - O(1)})$.

\subsection{An algorithm optimal in our framework}
As mentioned in Remark~\ref{rem:time_and_depth}, for fixed $c\ge 2$,
an optimal splitting algorithm, i.e. with smallest splitting depth
implies an optimal algorithm in $\mathfrak{A}_c$. Indeed, the
following lemma assures that such an optimal splitting algorithm can
be computed in quasi-polynomial time.

\begin{lemma}\label{lem:search_splitting_depth}
For a tree T with N nodes, $\mathsf{MSD}_{c}(T,\emptyset)$ can be
computed in time $O^{*}\left(N^{\mathsf{MSD}_{c}(T,
\emptyset)}\right)$ and polynomial space.
\end{lemma}
\begin{proof}
For any $h$, whether $\mathsf{MSD}_{c}(T, \emptyset)\leq h$ in
$O^{*}(N^{d})$ can be tested by a branch-and-bound algorithm. For any tree $T$ and
previous splitting nodes $S$, enumerate all nodes and check
whether the minimal splitting depth of each subtree is at most $h-1$
recursively. The maximal depth of the recursion is set to $h$, thus the
running time is $O^{*}\left(N^{h}\right)$. Therefore, the overall
time is at most $O^{*}\left(N^{\mathsf{MSD}_{c}(T, \emptyset)}\right)$. The
required space is polynomial.
\end{proof}

By applying the splitting algorithm in the previous lemma we obtain the following tradeoff algorithm.

\begin{theorem}\label{the:optimal_hybrid}
For any $c\geq 2$, any $\epsilon$ satisfying $0<\epsilon<0$, the
hybrid algorithm can be done in time
$O^{*}(d^{(1-\epsilon)\mathsf{MSD}_{c}(T, \emptyset)\tw(\phi)})$ and
space $O^{*}(2^{c\epsilon\tw(\phi)})$.
\end{theorem}

\section{Some remarks on the complexity of bounded width $\sat$}\label{sec:complexity}
In this section we separate the complexity of the tree-width parameterized from the 
path-width parameterized $\sat$ for the same width value, and we initiate the study of the incompressibility and non-sparsification of width-parameterized $\sat$ instances.

\paragraph{More preliminaries and notation}
For the first part we proceed by giving a machine characterization of
$\sattw(w(n))$; the problem of deciding $\sat$ where the CNF formula is given together
with a tree-decomposition of width $w(n)$, where $n$ is the input length.
\cite{papakonstantinou2009note} gives a machine characterization of $\satpw(w(n))$; the corresponding problem for a given path decomposition.
We define $\nl[r(n)]$ to be the class of problems decidable by a log-space machine
equipped with a read-only non-deterministic, polynomially long tape,
where the machine makes $\leq r(n)$ passes (as the head reverses) over the witness tape.
It was shown that $\satpw(r(n)\log n)$ is complete for $\nl[r(n)]$ under many-to-one
logspace reductions. Also under the ETH $\nl[\omega(1)]$ is incomparable to $\mathbf{P}$.
The ETH states that $3$-$\sat$ on $N$ variables cannot be decided in time $2^{o(N)}$.
Further study shows that under the ETH together with the assumption that NP is not contain in some fixed polynomial space bound, as $r(n)$ grows from constant to polynomial there is a
strict hierarchy of classes that grows from $\nl$ all the way up to $\np$.

We strengthen the assumption that $\np\not\subseteq \conp$ to $\nl[\omega(1)]\not\subseteq \conp$. In fact,
we assume further that
$$\nl[\omega(1)]\not\subseteq \conp\textsf{/poly}$$
A complexity theoretic study of this assumption is interesting on its own right, and it is left for future work. 
Here are some indications on it validity:
(i) the belief that $\np\not\subseteq \conp$ is because usually people think that
in fact the required certificate size blows up to exponential (not merely super-polynomial)
- i.e. some kind of exhaustive enumeration is required -
and (ii) given a non-uniform advice of polynomial size won't help either
(in particular, an easy extension of Karp-Lipton shows that if $\np\subseteq\conp\textsf{/poly}$ the polynomial
hierarchy collapses).

\subsection{Tree-width parameterized $\sat$ is harder than path-width parameterized $\sat$}
We characterize $\sattw$ in terms of non-deterministic space bounded machines, that run in time polynomial, and in addition they
are equipped with an unbounded stack \cite{cook1971characterizations}. $\nauxpda(s(n), t(n))$ is the class of decision problems decidable by such a machine in space $O(s(n))$ and time $O(t(n))$. $\sac(s(n), t(n))$ is defined as the semi-unbounded boolean circuit of size $O(s(n))$ and depth $O(t(n))$, in which the AND gates have constant fan-in and all negations are at the input level. We have that $\sac(\log{n}, n^{O(1)})=:\sac^1$, and $\textbf{NC}^1\subseteq\sac^1\subseteq\textbf{AC}^1$.

The following theorem is a new characterization of $\sattw$. The proof of this theorem is very different than the path-width characterization theorem of \cite{papakonstantinou2009note}.

\begin{theorem}\label{thm:sat_tw}
$\sattw(\log^l{n})$ is hard for $\nauxpda(\log^l{n}, n^{O(1)}), \forall l\in \mathbb{Z}^+$, under $O(\log^l{n})$-space may-to-one reductions.
\end{theorem}
\begin{proof}
By the characterization theorems from \cite{venkateswaran1987properties} and \cite{ruzzo1980tree}, it can be shown that $\nauxpda(\log^l{n},poly(n))\subseteq\sac(2^{O(\log^l{n})}, \log^l{n})$ by $O(\log^l{n})$ space bounded transformations. Now it suffices to construct a $\sattw(\log^l{n})$ instance $\phi$, given the $\sac(2^{O(\log^l{n})}, \log^l{n})$ circuit $C$ and an input $x$, such that $\phi$ is satisfiable if and only if $C$ evaluates to $1$ on $x$.

Our construction generalizes the observations in \cite{gottlob2001complexity}. Without loss of generality we assume that the circuit has the following \emph{normal form}:
\begin{enumerate}[(1)]
\item Fan-in of all AND gates is $2$.
\item The circuit is \emph{layered}.
\item The circuit is \emph{strictly alternating}, odd-layer gates are OR, even-layer gates are AND.
\item The circuit has an odd number of layers.
\item NOT gates only appear in the bottom layer.
\end{enumerate}

\begin{figure*}[ht]
\centering \subfloat[A semi-unbounded circuit with a proof tree
highlighted. NOT gates are
hidden]{\includegraphics[height=140pt]{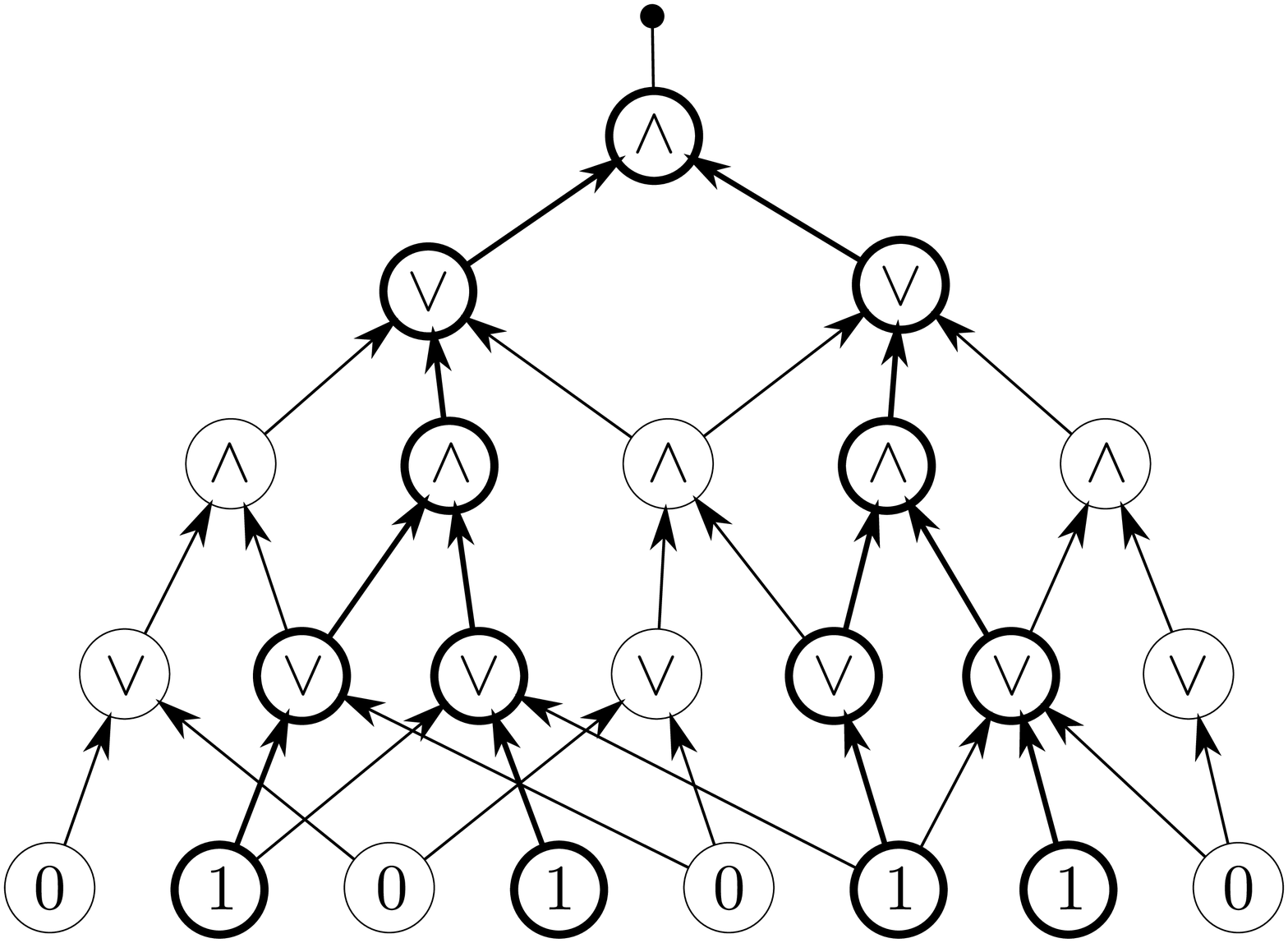}\label{subfig:proof_tree}}
\hspace{40pt} \subfloat[The skeleton of the proof trees]
{\includegraphics[height=140pt]{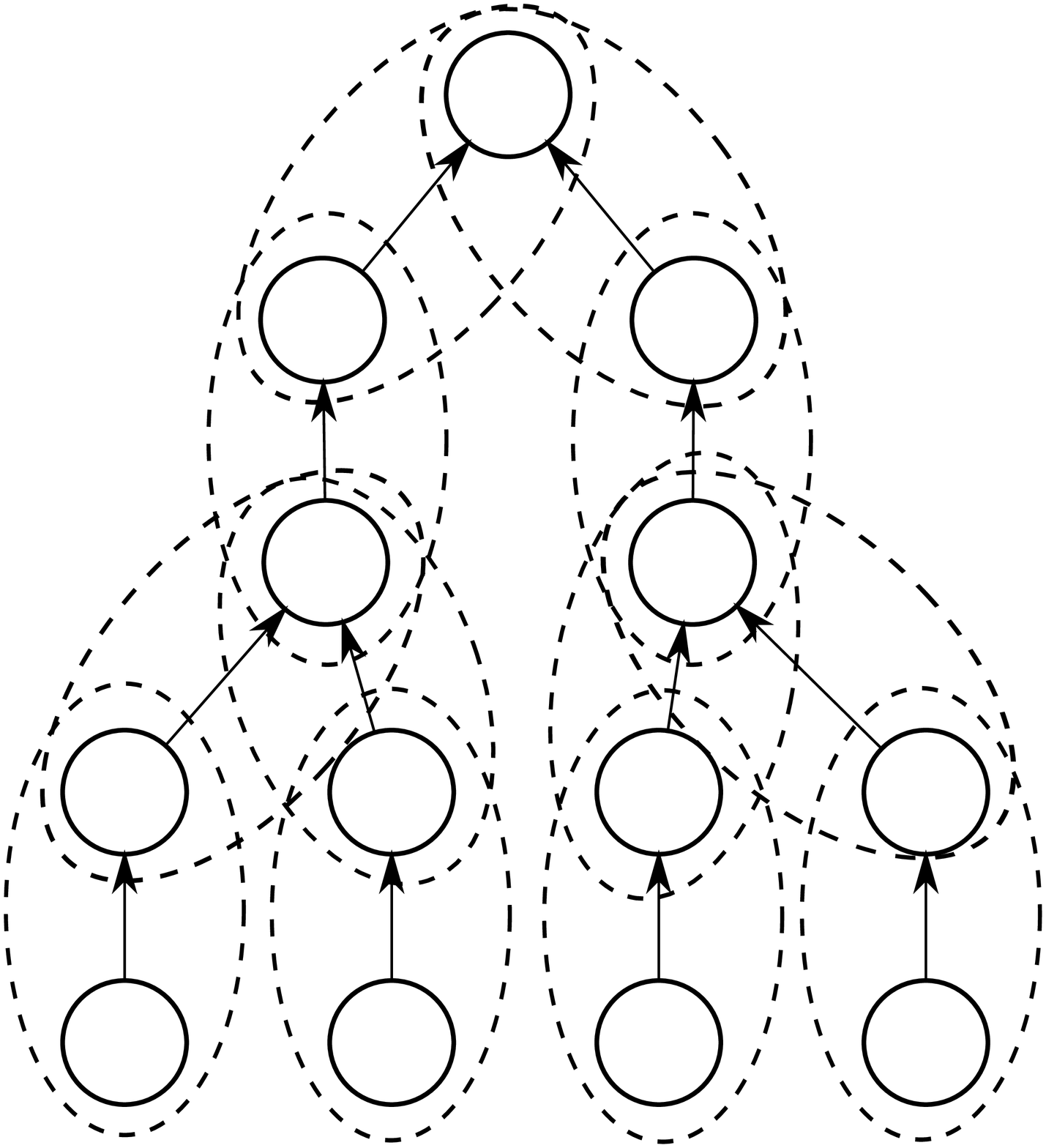}\label{subfig:skeleton}}
\caption{In
\ref{subfig:skeleton} a $\sattw(\log^l{n})$ instance is
constructed from the skeleton: each node corresponds to
$O(\log^l{n})$ boolean variables; clauses are constructed for each
dashed circle; and only those variables corresponding to a node
shared by different dashed circles must be put into a bag in the
tree decomposition, which ensures $O(\log^l{n})$ tree
width.}\label{fig:logcfl_hard}
\end{figure*}

A \emph{proof tree} is a tree with the same layering as the circuit. Each node of the tree is labeled by a gate from the corresponding layer of the circuit. At an odd layer, each node has one child, while at an even layer, each node has two children. Two connected nodes must be labeled such that the corresponding gates are connected. At the bottom layer, each node must be labeled by an input gate or a NOT gate which outputs value $1$. See Figure~\ref{subfig:proof_tree} for an example.

A proof tree can be viewed as the witness that a circuit evaluates to $1$ on the given input. One may observe that, since we assume that all the circuits are of normal form, every proof tree must have the same shape. The tree of the same shape of a proof tree without labeling is called a \emph{skeleton}(see Figure~\ref{subfig:skeleton}). Showing that the circuit evaluates to $1$ on an input is equivalent to giving a labeling satisfying the proof tree conditions to the skeleton.

One of the $2^{O(\log^l{n})}$ gates in the circuit can be indexed by a $O(\log^l{n})$ bit binary string. For each node $v$ in the skeleton, assign a variable $x_v$ using space $O(\log^l{n})$, to indicate the index of the gate that this node be labeled. This variable is also seen as a group of $O(\log^l{n})$ boolean variables. For each pair of connected nodes $u,v$ in the skeleton, and for each pair of possible indices $a,b$, which can be assigned to $x_u$ and $x_v$ correspondingly satisfying the conditions of a proof tree, create a clause encoding $x_u=a\wedge x_v=b$. All these clauses form an $\sat$ instance whose incidence graph has a tree decomposition with tree width $O(\log^l{n})$ (See Figure~\ref{fig:logcfl_hard} for an example and illustration).
\end{proof}

\begin{corollary}
$\sattw(\log n)$ is hard for $\sac^1$, under log-space many-to-one reductions.
\end{corollary}

This corollary follows by the characterization in \cite{venkateswaran1987properties}.

Therefore, under the standard assumption that $\nl\subsetneq \sac^1$ we separate the complexity of $\satpw(\log n)$ and $\sattw(\log n)$. How about higher width values? In this case we do not have well-established higher analogs of $\nl\subsetneq \sac^1$. However, the characterization of Theorem \ref{thm:sat_tw} together with the main theorem in \cite{papakonstantinou2009note} can be understood in the other direction. This means that the conjectures about the corresponding complexity classes may be true exactly because one may  conjecture that $\sattw(w(n))$ is harder than $\satpw(w(n))$. It seems that there is much more to be done towards this direction.

\subsection{Incompressibility and non-sparsification}
By compressibility of $\sat$ instances we mean that the input instance or parameters can be reduced by an efficient algorithm, in a way that the compressed instance preserves the satisfiability. Let us start with some preliminary observations stating that no non-trivial compression can be done to reduce the width parameter.

Suppose we have an instance $\phi$ together with an optimal tree decomposition of width $\tw(\phi)=\omega(\log|\phi|)$, and assume that there is a procedure running in polynomial time, which constructs a new instance $\phi'$ with tree-width $\tw(\phi')=\frac{1}{2}\tw(\phi)$ such that $\phi\iff \phi'$. If we repeat this procedure for $\log\tw(\phi)$ times, we will obtain an instance $\overline{\phi}$, where $\tw(\overline{\phi})$ is constant and has the same satisfiability as $\phi$. Satisfiability of $\overline{\phi}$ can be determined in polynomial time, and the transformation can also be done in polynomial time, therefore we can determine satisfiability of $\phi$ in polynomial time. Under ETH, this is not possible, which implies that determining satisfiability of $\phi$ requires $2^{\Omega(\tw(\phi))}$ time. Namely, such a procedure cannot exist.

Next we turn to the question of interactively ``compressing'' the instance length. The literature refers to this process as \emph{sparsification}. \cite{DM10} shows that in order for a polynomial-time bounded machine to decide 3-$\sat$ with the help of an unbounded oracle, if the number of bits needed to be sent to the oracle is $O(n^{3-\epsilon})$, where $n$ is the number of variables and $\epsilon>0$, then $\np\subseteq$ co$\np$\textsf{/poly}. Let $L$ be a language, denote $\mathbf{OR}(L)$ be the problem of asking whether one of the input instance belongs to $L$. The following lemma is crucial for the proof.

\begin{lemma}[\cite{DM10}]\label{lem:distillation}
Let $L$ be a language, with instance size $s$ and $t:\mathbb{Z}^+\to\mathbb{Z}^+$ be polynomially bounded s.t. the problem of $\mathbf{OR}(L)$ with $t(s)$ instances can be decided by sending $O(t(s)\log{t(s)})$ bits, then $L\in\conp$\textsf{/poly}.
\end{lemma}

Applying the same technique, we obtain the non-sparsification of 3-$\satpw(w(n))$ instances, $w(n)=\Omega(\log n)$.

\begin{lemma}
If a 3-$\satpw(w(n)\log n)$ instance can be decided by sending $O(n^{1-\epsilon})$ bits to the oracle, then $\nl[w(n)/\log n]\subseteq\conp$\textsf{/poly}.
\end{lemma}

\begin{proof}
Obviously, $s$ is at most polynomial in $n$. Consider an $\mathbf{OR}(3$-$\satpw(w(n)))$ instance, which contains $t(s)$ 3-$\satpw(w(n))$ instances each with $n$ variables can be represented by a 3-$\satpw(w(n))$ instance with $t(s)s$ variables. Assume all the instances use different variables. Suppose the instances are $\phi_i$, $\forall i$, and each has a corresponding path-decomposition $\mathcal{P}_i$, variables $v_{i,j}$, clauses $C_{i,j}$. Simply joining path-decompositions will impose an AND-relation instead of an OR-relation. To impose an OR-relation, first join the path-decompositions sequentially, let $a$ be a selector, for each $\mathcal{P}_i$, replace each clause $C_{i,j}$ by a clause representing $(a=i)\rightarrow C_{i,j}=\overline{(a=i)}\vee C_{i,j}$. $a$ can be implemented by $O(\log t(s))=O(\log n)$ variables appearing in every bag. One last step is that each newly created clause is of $O(\log n)$ variables, to break each of them into clauses of $3$ variables by standard technique, $O(w(n)\log n)$ variables need to be introduced. In the end, we constructed a 3-$\sat$ instance with path-width $w(n)\log n$. Now by hypothesis, this instance of $t(s)s$ variables can be decided by sending $(t(s)s)^{1-\epsilon}$ bits, by Lemma~\ref{lem:distillation}, this means 3-$\satpw(w(n))$ is in $\conp$\textsf{/poly}, and by the characterization mentioned before, the lemma follows.
\end{proof} 

\section{Conclusions}\label{sec:conclusion}
We devised a simple algorithm for deciding the satisfiability of arbitrary
CNF formulas, that runs in time $2^{O(\tw(\phi)\log|\phi|)}$ and space polynomial.
We conjecture that doing asymptotically better in the exponent blows up
the space to exponential in the tree-width. Our main technical development is a family of deterministic
algorithms that achieve tradeoffs by trading constants in the exponent between
space and running time. 

One issue we did not discuss is whether randomness can be helpful towards 
better algorithms for tree-width bounded $\sat$. Take for example Schoening's 
algorithm  \cite{schoning1999probabilistic} for $3$-$\sat$, 
which achieves running time $\alpha^n$, for a constant $\alpha<2$ and space polynomial. 
Intuitively, this algorithm is oblivious to any structure the given CNF may
have. An interesting question is finding a way to exploit the small width structure  
using randomness. 

This paper also initiates an in-depth study of the computational complexity 
of width-parameterized $\sat$. One possible direction is understanding the implications
of our conjecture; i.e. developing machinery that leads to either proving 
or disproving it. Another interesting research direction is to obtain stronger sparsification 
results than those obtained in Section \ref{sec:complexity}. The main technical obstacle
towards this goal is that the packing lemma of \cite{DM10} does not apply 
in the setting of bounded-width $\sat$. The reason is that even the first step in the
construction of \cite{DM10} blows up the tree-width from any value (e.g. $\log^2 n$) 
to linear. That is, following their work the given formula is right away transformed into
a formula of huge tree-width. We believe that proving non-sparcification $O(n)$ is possible,
and it seems to require the development of new non-sparsification tools. 

In general, understanding various aspects of the computational complexity of width-parameterized 
$\sat$ is an issue left open for future research.

\section*{Acknowledgments}
We would like to thank Kevin Matulef for useful remarks and suggestions.

\bibliographystyle{alpha}
\bibliography{ksat-bw}


\end{document}